\documentclass[12pt]{article}

\usepackage[backend=bibtex,style=ieee,
citestyle=numeric-comp,doi=false,date=year,isbn=false,url=false,eprint=false]{biblatex}
\usepackage{titlesec}

\titleclass{\subsubsubsection}{straight}[\subsection]

\newcounter{subsubsubsection}[subsubsection]
\renewcommand\thesubsubsubsection{\thesubsubsection\alph{subsubsubsection}}

\titleformat{\subsubsubsection}
  {\normalfont\normalsize\bfseries}{\thesubsubsubsection}{1em}{}
\titlespacing*{\subsubsubsection}
{0pt}{3.25ex plus 1ex minus .2ex}{1.5ex plus .2ex}

\makeatletter
\renewcommand\paragraph{\@startsection{paragraph}{5}{\z@}%
  {3.25ex \@plus1ex \@minus.2ex}%
  {-1em}%
  {\normalfont\normalsize\bfseries}}
\renewcommand\subparagraph{\@startsection{subparagraph}{6}{\parindent}%
  {3.25ex \@plus1ex \@minus .2ex}%
  {-1em}%
  {\normalfont\normalsize\bfseries}}
\def\toclevel@subsubsubsection{4}
\def\toclevel@paragraph{5}
\def\toclevel@paragraph{6}
\def\l@subsubsubsection{\@dottedtocline{4}{7em}{4em}}
\def\l@paragraph{\@dottedtocline{5}{10em}{5em}}
\def\l@subparagraph{\@dottedtocline{6}{14em}{6em}}
\makeatother

\usepackage{appendix}
\usepackage{amssymb}
\usepackage{amsthm}
\usepackage{amsmath}
\usepackage{float}
\usepackage{comment}
\usepackage{longtable}

  \usepackage[subnum]{cases}
  \usepackage{enumitem}
  
  \usepackage[table]{xcolor}
\usepackage{lipsum}
\usepackage{titlesec}
\usepackage{empheq}

\titlespacing\section{0pt}{0pt}{0pt}
\titlespacing\subsection{0pt}{0pt}{0pt}
\titlespacing\subsubsection{0pt}{0pt}{0pt}
\usepackage{mathtools}
\usepackage{multirow}
\usepackage{varwidth}
\usepackage{subcaption}
\usepackage{graphicx}
\usepackage[font=footnotesize,labelfont=bf]{caption}
\usepackage{empheq}
\usepackage[citecolor=blue,linkcolor=blue,colorlinks=true]{hyperref}%
 \newcommand\ackname{Acknowledgements}
     {\par\vfil\null\endtitlepage}

\newtheorem{theorem}{Theorem}[section]

\newtheorem{corollary}{Corollary}[section]

  {
      \theoremstyle{plain}
      
  }

\newcommand{\coop}{\mathcal{C}}
\newcommand{\defe}{\mathcal{D}}
\theoremstyle{definition}

\newtheorem{remark}{Remark}[section]
\numberwithin{equation}{section}

\usepackage[margin=1in, top=1in]{geometry}
\begin{document}

\title{Coupling the socio-economic and ecological dynamics of cyanobacteria: single lake and network dynamics}

\author{Christopher M. Heggerud\footnote{Corresponding author: Department of Mathematical and Statistical Sciences, University of Alberta, Edmonton, AB, Canada. Email: cheggeru@ualberta.ca},\ \ Hao Wang\footnote{Department of Mathematical and Statistical Sciences, University of Alberta, Edmonton, AB, Canada. Email: hao8@ualberta.ca},\ \ Mark A. Lewis\footnote{Department of Mathematical and Statistical Sciences and Department of Biological Sciences, University of Alberta, Edmonton, Alberta, Canada. Email: mark.lewis@ualberta.ca}}
\maketitle

\begin{abstract}
    In recent decades freshwater lakes have seen an increase in human presence. A common byproduct of this human presence is eutrophication, which readily results in harmful cyanobacteria blooms. In this work we propose a model that couples the socio-economic and ecological dynamics related to cyanobacteria systems. The socio-economic dynamics considers the choices a human population makes regarding whether or not to mitigate their pollution levels. These choices are based on various costs related to social ostracism, social norms, environmental concern and financial burden. The coupled model exhibits bistable dynamics, with one stable state corresponding to high mitigation efforts and low CB abundance, and the other to low mitigation efforts and high CB abundance. Furthermore, we consider social interactions among a network of lakes and present dynamic outcomes pertaining to various associated costs and social situations. In each case we show the potential for regime shifts between levels of cooperation and CB abundance. Social ostracism and pressure are shown to be driving factors in causing such regime shifts.

\end{abstract}

\section{Introduction}
Cyanobacterial harmful algal blooms (CHABs) are an ever present global concern in aquatic environments. The presence of CHABs often leads to several adverse outcomes both ecologically and economically. For example, CHABs can decrease ecosystem productivity by creating anoxic conditions and producing toxins as metabolic byproducts~\cite{Orr1990,Kaebernick2001}. Economically, CHABs add costs to water treatment, lower recreational and tourism value, and add risks when using freshwater for agricultural purposes. Although CHABs occur for a variety of reasons they are most commonly a result of eutrophication. Eutrophic conditions occur when an excess amount of nutrients required for organismal growth is in an aquatic ecosystem. Furthermore, eutrophication often occurs as a result of anthropogenic nutrient pollution from agriculture, industrial and urban run-off~\cite{Paerl2014}. In this sense there is a noteworthy connection between anthropogenic nutrient pollution and economic costs due to CHABs.

The study of systems where human and environmental dynamics are intertwined is beginning to receive more attention in the literature. For example, the importance of linking human and social dynamics to climate models to understand climate trajectories has been addressed~\cite{Beckage2020,Bury2019}. Other researchers have used social processes to better understand disease outbreaks~\cite{Pedro2020,Fair2021}. Ecologically, social dynamics have been coupled to forestry, fishery and other common-pool resource models to gain insight towards the balance between sustainable resource use and profit seekers~\cite{Satake2007,Farahbakhsh2021,Lee2011,Wang2016a}. Coupled socio-economic and ecosystem models for lake eutrophication have been considered by~\textcite{Iwasa2007,Iwasa2010}, but do not consider phytoplankton dynamics. In essence, human activities often result in changes in the ecological system, however changes in the ecological system will, in-turn, have an impact on the human behaviours thus creating a feedback loop. These types of systems are thought of as an integration between an ecological system and socio-economic system. Mathematical modelling of such systems typically involves the coupling of an ecological model that has terms dependent on human decisions to a human socio-economic model with outputs dependent on the state of the ecology~\cite{Iwasa2007,Satake2007}. 

Socio-economic models can be derived by considering social norms and pressures, monetary costs and psychology associated with the ecological system~\cite{Fransson1999}. As is the case in many current environmental issues, social ostracism can occur when an individual does not behave in a way that is environmentally favourable~\cite{Poon2015}. Social ostracism happens when a group or individual excludes or slanders another group or individual based on an action, opinion or response. Psychologically, being ostracised is harmful as humans have a basic want of being accepted~\cite{Williams2007}. As a response to ostracism humans often change behaviour to further avoid ostracism~\cite{Williams2011}. In the context of environmental issues, such as lake pollution, groups who assume non-environmentally favourable roles are often ostracised more than those that do~\cite{Poon2015,Iwasa2007,Sun2020} adding costs to the defection role. This means that modelling of socio-economic systems should include factors that account for social pressures. In addition, social norms often influence a person to assume a strategy regardless of its environmental impacts~\cite{Kinzig2013}. Socio-economic dynamics may be dependent on the frequency of each strategy, and not on the costs alone. Furthermore, the social costs due to ostracism and adherence to norms can be non-local and come from distanced social connections. Costs associated with pro-environmental roles often exceed the non-environmentally favourable role. These costs are often monetary and involve the investment in infrastructure to filter or treat urban water run-off. Additionally, lakes with low water quality and persistent HABs face additional costs associated with decreased land value, recreation, tourism based on the presence of toxins, and the visually and olfactorily unpleasant nature of HABs~\cite{Nicholls2018,Wolf2017}.

In many cases socio-economic models often have a game-theoretic component in which players choose one of several strategies based on the associated utility differences to the other strategies~\cite{Iwasa2007,Farahbakhsh2021,Suzuki2009,Iwasa2010,Sun2020}. Each strategy then has an associated disturbance of the ecological system, i.e. high vs. low pollution or deforestation rates. Individuals assume strategies at rates that are dependent on the perceived costs of each strategy, or fitness in game theory literature, and can be modelled in many different forms. For example, the logit best-response dynamics assumes there is a probability an individual assumes a strategy, where as the replicator dynamics assumes that the individual will always pick the most advantageous strategy~\cite{Sun2021,Farahbakhsh2021,Bury2019,Iwasa2007}. By explicitly considering distinct strategies and their associated costs ecosystem managers can use these models to gain insight towards policy implementation to obtain a favourable outcome.

 Many phytoplankton models have been used for the study of algal dynamics and take various forms including discrete time models, ODEs and PDES. In this study we extend a stoichiometric model that has been well established in the literature~\cite{Heggerud2020,Wang2007,Berger2006}. Ecological stoichiometry is defined as the study of the balance of energy and resources in ecological systems~\cite{Sterner2002}. This is a powerful tool as it allows the study of large scale phenomena, like CB abundance, by considering small scale components like internal energy and nutrients. The use of ecological stoichiometry has become increasingly common because of its ability to mechanistically capture the effects of resource limitations on ecological systems. For example, ecological stoichiometry has been used to study predator prey systems~\cite{Mitra2005,Branco2018}, producer-grazer systems~\cite{Wang2008,Loladze2000}, phytoplankton dynamics~\cite{Klausmeier2004,Wang2007}, toxicology~\cite{Peace2021}  and plant-disease dynamics~\cite{Lacroix2017} with great success.
 Ecological stoichiometry has been used to discuss the timescale separation between nutrient uptake and both algal growth and available nutrient depletion in~\textcite{Heggerud2020}. Separation of timescales allowed for the in-depth study of algal transient dynamics and driving mechanisms. This, along with many other studies, has established a solid modelling framework for phytoplankton dynamics~\cite{Wang2007,Berger2006,Huisman1994}. Additional complexity arises when coupling such ecological models to socio-economic models, both mathematically and in terms of timescales~\cite{Hastings2016,Hastings2010}. Human behaviour may change slower than the ecological dynamics and furthermore, the response of the ecological systems to human management strategies may be delayed~\cite{Carpenter2005,Hastings2016}. 
 
 Phosphorus is commonly considered to be a nutrient of interest in aquatic systems~\cite{Carpenter2005,Whitton2012}. Furthermore, the Redfield ratio (C:N:P=106:16:1)~\cite{redfield1934} implies that CB demands phosphorus more than other elements, except perhaps nitrogen~\cite{Sterner2002,Whitton2012}. However, since the demand for phosphorus is high the uptake rates and cell quotas for phosphorus will also be larger than other elements, expect perhaps nitrogen, and thus the corresponding phosphorus dynamics in the media occur on similar timescales to other ecological processes~\cite{Whitton2012,Heggerud2020}. Other nutrients, such as iron, can limit phytoplankton growth in a significant way by limiting photosynthesis, such as the case of peat lakes in the Netherlands~\cite{Smolders1993} and regions of the Antarctic~\cite{Koch2019}. The extended Redfield ratio implies the requirement of iron is much less than phosphorus and as a result cell quota values are small compared to those for phosphorus~\cite{Cunningham2017}. This means that the iron dynamics in the media may occur on a different timescale than the remaining ecological dynamics~\cite{Wurtsbaugh1983}. Thus, the timescale of the ecological dynamics depends on the study species and the nutrient being considered as uptake and growth rates can vary among species and nutrient~\cite{Whitton2012}.

In this paper we couple the ecological dynamics of cyanobacteria (CB) with the socio-economic dynamics of humans at each lake. We consider a network of lakes which are connected via social interactions only, allowing for presence of social norms and ostracism to influence human decision making. The ecological dynamics are given by extending the well established stoichiometric CB model of~\cite{Heggerud2020,Wang2007}. The socio-economic model is an extension of the models discussed in~\cite{Iwasa2007,Suzuki2009,Iwasa2010,Sun2020} in which individuals in a population choose to either cooperate by lowering pollution rates, or defect, by continuing to pollute at higher rates. The individuals choose their strategy based on costs associated with social pressure, concern for CB, tourism and recreation value, and infrastructure investment~\cite{Iwasa2007}. We fully derive the network model and offer several useful simplifications in Section~\ref{sec:model}. Our analysis begins in Section~\ref{sec:singlelake} where we consider the coupled dynamics at a single lake. The analysis of the single lake case is done by utilizing the separation in time scales in several different ways, including a phase line analysis for when phosphorus is the polluting nutrient in Section~\ref{sec:phase line} and phase plane analysis when iron is the polluting nutrient in Section~\ref{sec:phase plane}. In each case we observe bistable behaviour and gain insight towards the socio-economic parameter regions that lead to favourable outcomes. Lastly, in Section~\ref{sec:Networkmodel}, we revisit the network model . We simplify the network model to allow the system to be studied in the restricted phase plane showing three possible equilibria corresponding the low, high, and mixed levels of cooperation regimes throughout the network. Finally, discuss several two-parameter bifurcation plots which highlight under which parameter regions each regime occurs.

\section{A coupled cyanobacteria-socio-economic\\ network model}\label{sec:model}

In this section we extend a well established CB model~\cite{Wang2007,Heggerud2020,Berger2006} to account for socio-economic dynamics that alter the amount of anthropogenic nutrient input.
The CB model considers three state variables: CB abundance, cell quota, and available nutrient. The socio-economic component tracks the proportion of cooperators given by the best-response dynamics~\cite{Iwasa2007}. We separately consider phosphorus and iron as the limiting nutrient and introduce the phosphorus in this section. Several approximations of certain mechanistic modelling components are provided to aid in later analysis.   
 We assume that several distinct lakes are connected via social connections, due the presence of social communication. Each individual in the network assumes one of two strategies, cooperation or defection denoted with $\mathcal{C}$ and $\mathcal{D}$, respectively.  Locally, each strategy will face costs associated with the abundance of CB but only defectors will face a cost associated with social ostracism. In addition, we assume that each strategy faces a societal cost from the lake network that is proportional to the frequency of players of opposing strategies, this is referred to as a network social norm cost. 

We now couple a socio-economic model~\cite{Iwasa2007,Iwasa2010} a stoichiometric phytoplankton model~\cite{Wang2007,Heggerud2020,Berger2006} yielding

\begin{equation}\label{eq:CBnetwork}
    \begin{dcases}\left.
       \begin{aligned}
       \frac{dB_i}{dt}&=rB_i(1-\frac{Q_m}{Q_i})h(B_i)-\nu_r B_i-\frac{D}{z_e}B_i,
       \\
 \frac{dQ_i}{dt}&=\rho(Q_i,P_i)-rQ_i(1-\frac{Q_m}{Q_i})h(B_i),
\\
\frac{dP_i}{dt}&=\frac{D}{z_e}(I(F_i(t))-P_i)-B_i\rho(P_i,Q_i),
\\
\frac{dF_i}{dt}&=r_{i,\mathcal{D} \mathcal{C}}(F_i,B_i)(1-F_i)-r_{i,\mathcal{C} \mathcal{D}}(F_i,B_i)F_i,
       \end{aligned} \right. 
    \end{dcases}
\end{equation}
where $B_i,Q_i,P_i$ and $F_i$ represent the concentration of CB carbon biomass, the internal phosphorus to carbon nutrient ratio (cell quota), dissolved mineral phosphorus and the frequency of cooperators, respectively at lake $i$. The functions $h(B)$ and $\rho(Q,P)$ represent the light dependent growth of CB and phosphorus uptake, respectively. Both functions follow the form of~\cite{Heggerud2020,Wang2007} with
\begin{multline}\label{eq:h(B)}
     h(B)=\dfrac{1}{z_m}\int_0^{z_m}\dfrac{I_{in}\text{exp}{[-(K_{bg}+kB)s]}}{H+I_{in}\text{exp}{[-(K_{bg}+kB)s]}}ds
    \\
    =\dfrac{1}{z_m(K_{bg}+kB)}\text{ln}\left(\dfrac{H+I_{in}}{H+I_{in}\text{exp}{[-(K_{bg}+kB)z_m]}}\right),
\end{multline}
   
and
\begin{equation}
    \rho(Q,P)=\rho_m\dfrac{Q_M-Q}{Q_M-Q_m}\dfrac{P}{M+P}.
\end{equation}

The anthropogenic phosphorus addition is given as
 \begin{equation}
    I(F_i(t))=p_{i,\mathcal{D}}(1-F_i(t))+p_{i,\mathcal{C}}F_i(t),
\end{equation}
where $p_{i,\mathcal{C}}$ and $p_{i,\mathcal{D}}$ are the phosphorus input concentrations of the cooperators and defectors, respectively, with $p_{i,\mathcal{C}}\leq p_{i,\mathcal{D}}$.

The derivation of $h(B)$ is based on sound principles and assumptions of algal growth rates and light attenuation via the Lambert-Beer law. However, as with many other mathematical models, approximations of complex but meaningful functions can prove useful in analysis. We note that the key features of $h(B)$ are that it is monotone decreasing and that $\lim_{B\to\infty}h(B)=0$. Thus, we assume that $h(B)$ is sufficiently approximated as follows: 
\begin{equation}\label{eq:happ}
    h(B)\approx h_{app}(B)=\frac{1}{\tilde aB+\tilde b},
\end{equation}
where $\tilde a$ and $\tilde b$ are values such that $h(B)=h_{app}(B)$ for $B=0$ and $B=1/kz_m$ and are given as $\tilde b=1/h(0)$ and $\tilde a=\frac{kz_m }{h(1/kz_m)}-\frac{kz_m}{h(0)}$. The comparison of $h(B)$ and $h_{app}(B)$ is given in Figure~\ref{fig:happvsh}
\begin{figure}
    \centering
    \includegraphics[width=0.6\paperwidth]{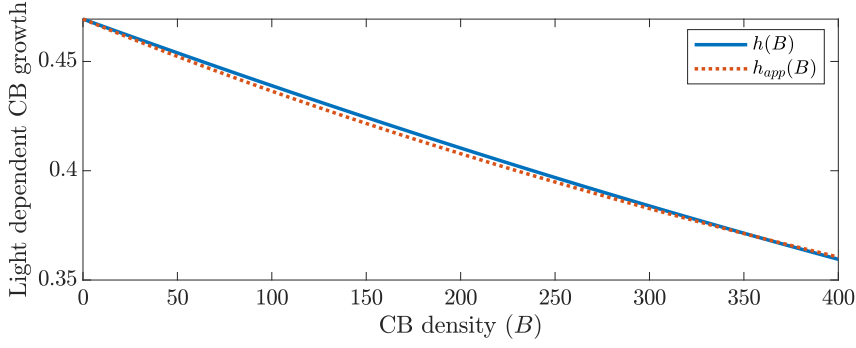}
    \caption{Light dependent growth function, $h(B)$, and its approximation, $h_{app}(B)$ given by~\eqref{eq:happ}. }
    \label{fig:happvsh}
\end{figure}

Many previous studies have established socio-economic dynamics based upon cost functions~\cite{Iwasa2007,Iwasa2010,Sun2021,Farahbakhsh2021,Satake2007}. We extend these results to suit our study in the following fashion. 
Let $C_{i,\mathcal{C}}(B_i)$ and $C_{i,\mathcal{D}}(F_i,B_i)$ denote the cost associated with each strategy at lake $i$. Each strategy has an associated baseline cost, $c_{i,\mathcal{C}}$ and $c_{i,\mathcal{D}}$ with $c_{i,\mathcal{C}}>c_{i,\mathcal{D}}$. Both strategies also face a `recreational' cost associated with the abundance of CB. Defectors face an additional cost of social ostracism that increases with CB abundance and the frequency of cooperators. Additionally, each strategy faces a network social norm cost that is proportional to the connectivity to each lake in the network and frequency of players of opposing strategy at that lake. The costs faced by the defector and cooperator are given respectively by  
\begin{multline}\label{eq:CD}
     C_{i,\mathcal{D}}(F_i,B_i)=\underbrace{c_{i,\mathcal{D}}}_\text{\scriptsize baseline cost}+\alpha\underbrace{(1+\xi F_i)}_\text{\scriptsize ostracism}\underbrace{\psi  B_i}_\text{concern for CB}
     \\+\underbrace{\phi B_i}_\text{\scriptsize cost of CB}+\underbrace{d_\defe\bar F,}_\text{social norm pressure}
\end{multline}
\begin{equation}\label{eq:CC}
    C_{i,\mathcal{C}}(B_i)=c_{i,\mathcal{C}} +\phi B_i +\underbrace{d_\coop(1-\bar F)}_\text{social norm pressure},
\end{equation}

where $\bar F=\dfrac{\sum_{ j}d_{ji}F_j(t)}{\sum_j d_{ji}}$ is the weighted average of the frequency of cooperators in the network. Further assume that the actual cost of each strategy is stochastic with a known cost and a random cost, given by $U_\coop=C_\coop+\epsilon_\coop$ and $U_\defe=C_\defe+\epsilon_\defe$. Since we are considering $\epsilon_\coop$ and $\epsilon_\defe$ to be additional random costs, their maximum values are of most interest. Thus, we assume that $\epsilon_\coop$ and $\epsilon_\defe$ follow the extreme value (Gumbel) distribution. Conveniently, the difference between two extreme value distributed random variables follows a logistic distribution~\cite{Hofbauer2003}. That is, $\epsilon_d=\epsilon_\coop-\epsilon_\defe\sim \textup{Logistic}(0,\frac{1}{\beta})$ with CDF $\dfrac{1}{1+e^{-\beta x}}$.  Thus, when a player evaluates their strategy they will defect with probability $P_\defe=P(U_\defe<U_\coop)=P(\epsilon_\defe-\epsilon_\coop<C_\coop-C_\defe)$ (or cooperate with probability $P_\coop=P(U_\defe>U_\coop)=P(\epsilon_\defe-\epsilon_\coop>C_\coop-C_\defe)$) given by the logistic distribution. 

Finally, the rate of switching is given as the probability of choosing a strategy, multiplied by the rate at which one evaluates their strategy: \begin{equation}\label{eq:rDC}
    r_{i,\mathcal{D} \mathcal{C}}(F_i,B_i)=\frac{s}{1+e^{\beta[C_{i,\mathcal{C}}(B_i)-C_{i,\mathcal{D}}(F_i,B_i)]}},
\end{equation}
\begin{equation}\label{eq:rCD}
   r_{i,\mathcal{C} \mathcal{D}}(F_i,B_i)=\frac{s}{1+e^{\beta[C_{i,\mathcal{D}}(F_i,B_i)-C_{i,\mathcal{C}}(B_i)]}},
\end{equation}
where $s$ is the level of conservatism of the population interpreted as the rate at which a player evaluates their strategy. If $s$ is small the population switches strategies infrequently.  $\beta$ is a parameter controlling the level of stochasticity. Large $\beta$ means the population deterministically chooses a strategy based on cost, whereas a small $\beta$ will make the switching more random as seen in Figure~\ref{fig:pwiseapp}. Furthermore, the last equation in~\eqref{eq:CBnetwork} can be written as
\begin{equation}
    r_{i,\mathcal{D} \mathcal{C}}(F_i,B_i)(1-F_i)-r_{i,\mathcal{C} \mathcal{D}}(F_i,B_i)F_i=r_{i,\defe\coop}(F_i,B_i)-sF_i.
\end{equation} The switching rates described in~\eqref{eq:rDC} and~\eqref{eq:rCD} arrive from a sound derivation and are quite intuitive and often referred to as the logit best-response model for choice probabilities. However, as previously discussed the approximation of complex functions by mathematically tractable functions is an incredibly useful tool. For this reason we note that the logistic function ($\frac{1}{1+e^{\beta x}}$) is readily approximated by the ramp function
\begin{equation}\label{eq:logapprox}
    \frac{1}{1+e^{-\beta x}}\approx  \begin{cases} 
      0 & x\leq -c^*, \\
      \frac{1}{2}+\tilde \beta x & -c^*\leq x\leq c^*, \\
      1 & c^*\leq x ,
  \end{cases}
\end{equation}
where $c^*=\frac{1}{2 \tilde \beta}$ and $\tilde\beta$ is a parameter found by minimising $\mathcal{L}^1$ norm of the difference between the two functions for a given value of $\beta$. Thus, $r_{i,\defe\coop}(F_i,B_i)$ is approximated by  
\begin{equation}\label{eq:rapproxfull}
    \hat r_{i,\defe\coop}(F_i,B_i)=s\cdot\begin{cases} 
      0 & 1/2+\tilde\beta (C_{i,\defe}-C_{i,\coop})\leq 0, \\
      1/2+\tilde\beta (C_{i,\defe}-C_{i,\coop})  & 0<1/2+\tilde\beta (C_{i,\defe}-C_{i,\coop})<1 ,\\
      1 &1\leq 1/2+\tilde\beta (C_{i,\defe}-C_{i,\coop}).
  \end{cases}
\end{equation}
\begin{figure}
    \centering
    \includegraphics[width=0.6\paperwidth]{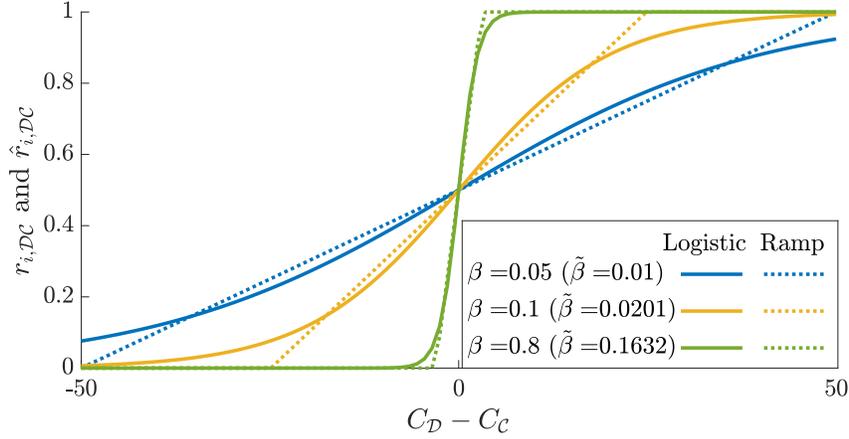}
    \caption{Shows the comparison of the logistic function (given in~\eqref{eq:rDC}) and the approximating ramp function in~\eqref{eq:rapproxfull} for various values of $\beta$ and the corresponding $\tilde\beta$ values. We take $s=1$ here. }
    \label{fig:pwiseapp}
\end{figure}
Further, the piecewise function can be equivalently written as
\begin{equation}\label{eq:rapproxfullminmax}
    \hat r_{i,\defe,\coop}(F_i,B_i)=s\cdot\textup{max}\{0,\textup{min}\{1,\frac{1}{2}+\tilde\beta (C_{i,\defe}-C_{i,\coop})\}\}.
\end{equation}
The difference between  $r_{i,\defe\coop}(F_i,B_i)$ and $\hat r_{i,\defe\coop}(F_i,B_i)$ is shown in Figure~\ref{fig:pwiseapp}.



The parameters and their values corresponding to the ecological components of model~\eqref{eq:CBnetwork} are summarized in Table~\ref{tab:ecoparamtable}. In-depth discussion and descriptions of the ecological parameters can be found in~\cite{Heggerud2020,Wang2008} and the references therein. The parameters corresponding to the socio-economic dynamics are summarized in Table~\ref{tab:socecparamtable} and are taken from the ranges in~\cite{Iwasa2007} and are justified by arguing a comparable scale of all terms in~\eqref{eq:CD} and~\eqref{eq:CC}. 
 
 \begin{table}[h]
\centering
\caption{Definitions and values for ecological parameters of system~\ref{eq:CBnetwork}}
{
\begin{tabular}{| V{2.1cm}| V{8.5cm} | V{1.5cm} | V{6.1cm} | V{2.5cm} |} 
  \hline
Par.  & Meaning & Value for sim. & Biological Values & Ref.
\\
\hline
$r$ & Maximum CB specific production rate& 1& 1 $/day$ &~\cite{Diehl2005}
\\ 
$Q_m$ & CB cell quota at which growth ceases (minimum)&0.004& 0.004 $gP/gC$&~\cite{Diehl2005}
\\
 $Q_M$ & CB cell quota at which nutrient uptake ceases (maximum)&0.04& 0.04 $gP/gC$&~\cite{Diehl2005}
 \\
 $z_m$ & Depth of epilimnion&7& $>0-10m$&~\cite{Kalff2002}
\\
$\nu_r$ & CB respiration loss rate& 0.35& 0.05-0.6 $/day$&~\cite{Whitton2012,Berger2006}
\\
$D$ & Water exchange rate& 0.02& $m/day$&~\cite{Berger2006}
\\
$H$ & Half saturation coefficient of light-dependent CB production & 120& 120 $\mu mol/(m^2\cdot s)$&~\cite{Diehl2005}
\\
$\rho_m$ & Maximum CB phosphorus uptake rate& 1& 0.2-1 $gP/gC/day$ &~\cite{Berger2006,Diehl2005}
\\
$M$ & Half saturation coefficient for CB nutrient uptake& 1.5& 1.5 $mgP/m^3$&~\cite{Diehl2005}
\\
$K_{bg}$ & Background light attenuation &0.3& 0.3-0.9 $/m$&~\cite{Berger2006,Diehl2005}
\\
$k$ & Algal specific light attenuation &0.0004&  0.0003-0.0004 $m^2/mgC $&~\cite{Berger2006,Diehl2005}
\\
$I_{in}$ & Light intensity at water surface& 300& 300 $\mu mol/(m^2\cdot s)$&~\cite{Diehl2005}
\\
\hline
\end{tabular}}
\label{tab:ecoparamtable}
\end{table}



\begin{table}[h!t]
    \centering
      \caption{Definitions and values for the socio-economic parameters of system~\eqref{eq:CBnetwork}}
       {
       
    \begin{tabular}{| V{2.1cm}| V{9.5cm} | V{1.5cm} | V{6.1cm} |} 
        \hline
Par.  & Meaning & Value & Units
\\\hline
$p_{i,\mathcal{C}}$ &  Concentration of influx of dissolved inorganic phosphorus for strategy $\mathcal{C}$. & 50& $mgP/m^3$
\\
$p_{i,\mathcal{D}}$ &  Concentration of influx of dissolved inorganic phosphorus for strategy $\mathcal{C}$. & 770& $mgP/m^3$
\\
$s$ &  Rate players make a decision to change strategies. & 0.001& $day^{-1}$
\\
$\beta$ &  Level of determinism in changing strategies. & 0.1& $\textup{(cost unit)}^{-1}$
\\
$\tilde\beta$ &  Slope of approximated line in~\eqref{eq:logapprox} & 0.0201& $\textup{(cost unit)}^{-1}$
\\
$c_{i,\mathcal{C}}$ &  Baseline cost to cooperate. & 50 & $\textup{(cost unit)}$
\\
$c_{i,\mathcal{D}}$ &  Baseline cost to defect. & 1 & $\textup{(cost unit)}$
\\
$\phi$ &  Cost conversion coeff. for CB & 10 & $\textup{(cost unit)}/mgC/m^3$
\\
$\alpha$ &  Cost conversion for social pressure due to CB  & 3 & $\textup{(cost unit)}$
\\
$\xi$ &  Strength of frequency dependence for social pressure & 10 & $-$
\\
$\psi$ & Level of social concern for CB & 0.02 & $(mgC/m^3)^{-1}$
\\
$d_{ji}$ &  connectedness of lake $j$ to $i$. & - & -
\\
$d_\defe$ &  Cost conversion coeff. of social norms for defecting. & 1 & $\textup{(cost unit)}$
\\
$d_\coop$ &  Cost conversion coeff. of social norms for cooperating. & 1 & $\textup{(cost unit)}$
\\
\hline
    \end{tabular}}
    \label{tab:socecparamtable}
\end{table}

\section{Dynamics of a single lake model}\label{sec:singlelake}
In this section we consider the single lake version of model~\eqref{eq:CBnetwork} where the external network pressure is treated as a constant. We separately consider the dynamics under phosphorus limitation and iron limitation, proceeding with a phase line and phase plane analysis, respectively. In each case bistability scenarios arise and bifurcation results are obtained.

In this paper we assume that when considering phosphorus almost all ecological processes occur on a fast time scale, thus the QSSA reduces the model to a single equation that represents the human dynamics on the slow timescale. When iron is considered, only the cell quota and CB dynamics occur on the fast time scale thus, the QSSA reduces the model to two differential equations on the slow timescale that represent the human and available iron dynamics. Hence, two types of analysis are performed. First, we consider a phase line analysis for the phosphorus system in Section~\ref{sec:phase line}. Second, we perform a phase plane analysis for the iron system in Section~\ref{sec:phase plane}

To start, assume that all other lakes are in a  fixed state allowing us to drop the subscript $i$. Thus the cost difference $C_\mathcal{C}-C_\mathcal{D}= c_\mathcal{C}-c_\mathcal{D}-\alpha(1+\xi F)\psi B+d_\coop(1-\bar F)-d_\defe \bar F$ can be written as $c_C-c_\mathcal{D}-\alpha(1+\xi F)\psi B +\hat\delta$, where $\hat\delta$ is treated as a parameter. In this section we study the following model:

\begin{equation}\label{eq:nondimsinglelake}
    \begin{dcases}\left.
       \begin{aligned}
       \frac{dB}{dt}&=rB(1-\frac{Q_m}{Q})h(B)-\nu_r B-\frac{D}{z_e}B,
       \\
 \frac{dQ}{dt}&=\rho(Q,P)-rQ(1-\frac{Q_m}{Q})h(B),
\\
\frac{dP}{dt}&=\frac{D}{z_e}(I(F(t))-P)-B\rho(P,Q),
\\
\frac{dF}{dt}&=r_{\mathcal{D} \mathcal{C}}(F,B)(1-F)-r_{\mathcal{C} \mathcal{D}}(F,B)F=\frac{s}{1+e^{\beta(C_\coop-C_\defe)}}-sF.
       \end{aligned} \right. 
    \end{dcases}
\end{equation}

\subsection{Dynamics of the phosphorus explicit model}\label{sec:phase line}
In this section we simplify system~\eqref{eq:nondimsinglelake} and use parameter values given for the phosphorus system in Tables~\ref{tab:ecoparamtable} and~\ref{tab:socecparamtable}. The simplifications lead to a single differential equation that is analyzed on the phase line to gain in-depth understanding of the single lake dynamics and the bistable nature of the system.

\subsubsection{Nondimensionalization of the single lake model} We begin by nondimensionalizing system~\eqref{eq:nondimsinglelake} by
letting $\tau=rt$, $u=kz_mB$, $v=\frac{Q}{Q_M}$, $w=\frac{P}{M}$, and $F$ remains unchanged as $F$ is dimensionless by definition. Making these substitutions into system~\eqref{eq:nondimsinglelake} yields:
\begin{equation}\label{eq:nondim2}
    \begin{dcases}\left.
       \begin{aligned}
       \frac{du}{d\tau}&=u(1-\frac{Q_m}{Q_M}\frac{1}{v})h(au)-\frac{(\nu_r+\frac{D}{z_e})}{r} u
       \\
 \frac{dv}{d\tau}&=\frac{\rho_m}{rQ_M}\frac{Q_M-Q_Mv}{Q_M-Q_m}\frac{w}{1+w}-(v-\frac{Q_m}{Q_M})h(au),
\\
M\frac{dw}{d\tau}&=\frac{D}{rz_e}(p_\coop F+p_\defe(1-F)-Mw)-\frac{\rho_M}{rkz_e}u\frac{Q_M-Q_M v}{Q_M-Q_m}\frac{w}{1+w},
\\
\frac{dF}{d\tau}&=\frac{s}{r}\left(\dfrac{1}{1+e^{\beta(c_C-c_\mathcal{D}-\alpha(1+\xi F)\psi au +\hat\delta)}}-F\right).
       \end{aligned} \right. 
    \end{dcases}
\end{equation}
 Upon substitution of the nondimensional parameters given in Table~\ref{tab:ndparamtableM} we have:
\begin{equation}\label{eq:nondimbetaeqepsilon}
    \begin{dcases}\left.
       \begin{aligned}
       \frac{du}{d\tau}&=u(1-\frac{1}{\gamma}\frac{1}{v})\hat{h}(u)-(\epsilon\beta_1+\beta_2) u
       \\
 \frac{dv}{d\tau}&=\omega(1-v)\frac{w}{1+w}-(v-\frac{1}{\gamma})\hat{h}(u),
\\
\frac{dw}{d\tau}&=\epsilon(\kappa_1 F-\beta_1 w)+\kappa_2(1-F)-\lambda u(1-v)\frac{w}{1+w},
\\
\frac{dF}{d\tau}&=\epsilon\left(\dfrac{1}{1+e^{\eta-\sigma(1+\xi F)u)}}-F\right),
       \end{aligned} \right. 
    \end{dcases}
\end{equation}
\begin{table}[h!]
\caption{Dimensionless parameters for system~\eqref{eq:nondimbetaeqepsilon} and equation~\eqref{eq:approxDF}.}
\centering
 {
\begin{tabular}{|l |c| c |} 
\hline
Parameter & Definition & Value 
\\
\hline
$\beta_1$&$\dfrac{D}{s z_m}$ & 0.2857  
\\
$\beta_2$&$\nu_r/r$& 0.35 
\\
 $\omega$&$\dfrac{\rho_m}{r(Q_M-Q_m)}$& 5.556  
 \\
 $\gamma$& $\frac{Q_M}{Q_m}$ & 10
\\
 $\kappa_1$& $\frac{p_\mathcal{C}}{M}\beta_1$ & 9.5238
 \\
 $\kappa_2$& $\frac{p_\mathcal{D}}{M}\frac{D}{rz_m}$ & 1.4667
\\
$\lambda$&$\dfrac{Q_M}{Q_M-Q_m}\dfrac{\rho_m }{M    r k z_m}$&52.9
\\
$k_1$&$z_m K_{b g}$& 2.1
\\
$I$&$I_{in}/H$ & 2.5
\\
$\eta$&$\beta(c_\mathcal{C}-c_\mathcal{D}+\hat\delta)$& -5 to 7 
\\
$\hat\eta$&$\tilde\beta(c_\mathcal{C}-c_\mathcal{D}+\hat\delta)$& -1 to 1.5 
\\
$\sigma$&$\alpha\beta \psi / kz_m$& 2.1429 
\\
$\hat\sigma$&$\alpha\tilde\beta \psi / kz_m$& 0.4307
\\
$\epsilon$&$s/r$& $<$0.01
\\\hline
\end{tabular}}
\label{tab:ndparamtableM}
\end{table}
where \begin{equation}\label{eq:h(u)nondim}
    \hat{h}(u)=\frac{1}{u+k_1}\log\left(\frac{1+I}{1+I\exp(-u-k_1)}\right),
\end{equation} is the non-dimensional light dependent growth term from~\eqref{eq:h(B)} and its nondimensional approximation stemming from~\eqref{eq:happ} is given as
\begin{equation}\label{eq:hhatapp}
    \hat{h}(u)\approx\hat h_{app}(u)=\frac{1}{au+b},
\end{equation}
where $b=1/\hat h(0)$ and $a=\frac{1}{\hat h(1)}-b$.

\subsubsection{Application of the quasi steady state approximation}\label{sec:phase lineQSSA}
 We now further reduce the model by utilizing the QSSA. The nondimensional system~\eqref{eq:nondimbetaeqepsilon} contains the parameter $\epsilon=s/r$, where $s$ is given as the rate at which players reevaluate strategies and $r$ is the maximal growth rate of CB. The rate at which players are able to reevaluate their strategy is very small in comparison to many ecological processes. Here we assume that the ecological dynamics of the CB occur on the order of days or weeks, whereas the social dynamics, or the maximum rate a player can switch strategies, is on the order of several months, or years. Thus, $\epsilon$ is a small parameter. 

 By re-scaling time with the small parameter $\epsilon$ in system~\eqref{eq:nondimbetaeqepsilon} we apply the QSSA. We introduce a new time scale $\tilde \tau=\epsilon\tau$ creating a slow time scale. The time scale $\tilde\tau$ is the slow timescale in which the human ($F$) dynamics occur, while $\tau$ is the fast timescale where most of the ecological dynamics occur. We note that certain aspects of the ecological dynamics such as water exchange rates can also occur on the slow timescale. Upon re-scaling time to $\tilde\tau$ we arrive at the following system:

\begin{equation}
    \begin{dcases}\left.
       \begin{aligned}
     \epsilon  \frac{du}{d\tilde\tau}&=u(1-\frac{1}{\gamma v})\hat{h}(u)-(\epsilon\beta_1+\beta_2) u,
       \\
 \epsilon\frac{dv}{d\tilde\tau}&=\omega(1-v)\frac{w}{1+w}-(v-\frac{1}{\gamma})\hat{h}(u),
\\
\epsilon\frac{dw}{d\tilde\tau}&=\epsilon(\kappa_1 F-\beta_1 w)+\kappa_2(1-F)-\lambda u(1-v)\frac{w}{1+w},
\\
\epsilon\frac{dF}{d\tilde\tau}&=\epsilon \dfrac{1}{1+e^{\eta-\sigma(1+\xi F)u)}}-\epsilon F.
       \end{aligned} \right. 
    \end{dcases}
\end{equation}
Now, by the QSSA, which assumes that the fast dynamics are in an equilibrium state, and letting $\epsilon$ go to zero we arrive at the differential algebraic system:
  \begin{subequations}\label{eq:phase linefull}
        \begin{empheq}[left=\empheqlbrace]{align}
         \frac{dF}{d\tilde\tau}&= \dfrac{1}{1+e^{\eta-\sigma(1+\xi F)u}}- F ,\label{eq:phase lineDiff}
    \\
      0&=\kappa_2(1-F)-\lambda u(1-v)\frac{w}{1+w},\label{eq:algebraic3Dsystemw}
     \\
  0&=u(1-\frac{1}{\gamma v})\hat{h}(u)-\beta_2 u, \label{eq:algebraic3Dsystemu}
  \\
  0&=\omega(1-v)\frac{w}{1+w}-(v-\frac{1}{\gamma})\hat{h}(u). \label{eq:algebraic3Dsystemv}
  \end{empheq}
  \end{subequations}
 Denote $u^*(F)$ as the solution to the algebraic system~\eqref{eq:algebraic3Dsystemw}-\eqref{eq:algebraic3Dsystemv}. The following theorem ensures that there is a unique solution to the algebraic system for the given parameter values. 
\begin{theorem}\label{thm:exunphase line}
 There exists a unique positive solution to the algebraic system defined by equations~\eqref{eq:algebraic3Dsystemw} to~\eqref{eq:algebraic3Dsystemv}  if $(1-\frac{1}{\gamma})\hat h(\frac{\omega\kappa_2(1-F)}{\beta_2\lambda })-\beta_2>0$.
 \end{theorem}

 \begin{proof}

 First, by multiplying equation~\eqref{eq:algebraic3Dsystemw} by $\omega/\lambda$ and adding equation~\eqref{eq:algebraic3Dsystemu} multiplied by $v$ and equation~\eqref{eq:algebraic3Dsystemv} multiplied by $u$ we arrive at the equation:
 \begin{align}
    0&= \frac{\omega}{\lambda}\kappa_2(1-F)-\beta_2uv,
    \\
   \iff u&=\frac{\omega\kappa_2(1-F)}{\beta_2\lambda v} =G(v).\label{eq:G(v)}
 \end{align}
 Substituting $u=G(v)$ into equation~\eqref{eq:algebraic3Dsystemu} divided by $u$ gives
 \begin{equation}\label{eq:implicitustar}
 0= (1-\frac{1}{\gamma v})\hat h(G(v))-\beta_2=S(v).
 \end{equation}
$G(v)$ is a decreasing function of $v$ and furthermore, recall that $\hat h(u)=\frac{1}{u+k_1}\log\left(\frac{1+I}{1+I\exp(-u-k_1)}\right)$ is a decreasing function of $u$, by construction. Thus, $S(v)$ is a strictly increasing function of $v$ which guarantees uniqueness.
Now, by construction of the biological system, $v\in[\frac{1}{\gamma},1]$ and $S(1/\gamma)<0$. Thus, if $S(1)=(1-\frac{1}{\gamma})h(\frac{\omega\kappa_2(1-F)}{\beta_2\lambda })-\beta_2>0$ then by the intermediate value theorem a solution to~\eqref{eq:implicitustar} exists. Lastly, equations~\eqref{eq:algebraic3Dsystemw} and~\eqref{eq:algebraic3Dsystemv} yield linear equations in $w$ ensuring uniqueness.

 \end{proof}
 
 \begin{remark}
  Theorem~\ref{thm:exunphase line} also applies when using $\hat h_{app}(u)$ in place of $\hat h(u)$. The condition for existence and uniqueness of a positive solution remains the same and an explicit form of $u^*(F)$ can be obtained. 
 \end{remark}
 \begin{remark}
 When the condition in Theorem~\ref{thm:exunphase line} is not satisfied a unique trivial solution can only exist when $F=1$. Otherwise, no positive solution exists. 
 \end{remark}
 
\subsubsection{An approximation for the cyanobacteria abundance} \label{sec:ustarF}

In this Subsection we explicitly compute $u^*(F)$ by utilizing the previously established approximation for $\hat h(u)$, given  in~\eqref{eq:hhatapp}, and solving~\eqref{eq:algebraic3Dsystemw}-\eqref{eq:algebraic3Dsystemv}.
Without approximations or simplifications the unique positive solution to~\eqref{eq:algebraic3Dsystemw}-\eqref{eq:algebraic3Dsystemv} is verified to exist by Theorem~\ref{thm:exunphase line} but can only be implicitly given.

Proceeding, from equation~\eqref{eq:G(v)}:
\begin{align}
    v&=\frac{\omega\kappa_2(1-F)}{\lambda\beta_2u}=p_1\frac{(1-F)}{u}.\label{eq:v(u)3Dphase line}
\end{align} We explicitly solve for $u^*(F)$ by utilizing the approximation for $\hat{h}(u)$ given by $\hat{h}_{app}(u)$ in~\eqref{eq:hhatapp}, using $v$ as in~\eqref{eq:v(u)3Dphase line}, and solving~\eqref{eq:algebraic3Dsystemu} for $u$ gives 
\begin{equation}\label{eq:ustar(F)}
    u^*(F)=\frac{\gamma p_1(1-\beta_2b)(1-F)}{\beta_2a\gamma p_1(1-F)+1}=\frac{a_1(1-F)}{a_2(1-F)+1},
\end{equation}
where $a_1=\gamma p_1(1-\beta_2b)>0$ and $a_2=\beta_2a\gamma p_1>0$.

\begin{remark}
 As shown in Figure~\ref{fig:ustarFvsustar} the explicit version of $u^*(F)$, as in~\eqref{eq:ustar(F)}, is a reasonable approximation to the numerical solution of~\eqref{eq:algebraic3Dsystemw}-\eqref{eq:algebraic3Dsystemv}.  Note that both solutions give $u^*(1)=0$. However, in reality even with 100\% cooperation we would predict a small but non-zero CB abundance due to the non-zero pollution rate of the cooperators. In our QSSA this term ($\kappa_1F$) disappears, and is hence essentially deemed negligible resulting in $u^*(1)=0$. 
\end{remark}

\begin{figure}
    \centering
    \includegraphics[width=0.6\paperwidth]{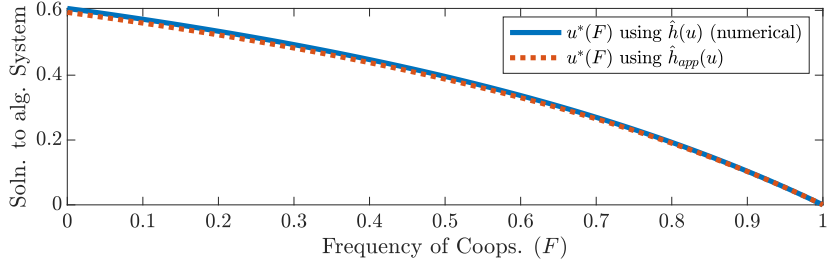}
    \caption{Comparison of the approximation given by~\eqref{eq:ustar(F)} and the numerical solution for $u^*(F)$ }
    \label{fig:ustarFvsustar}
\end{figure}
\subsubsection{A simplifying approximation for the governing differential equation}
We now apply the approximation discussed in~\eqref{eq:rapproxfullminmax} to~\eqref{eq:phase lineDiff} to solve for equilibrium values. 
From Section~\ref{sec:ustarF} we obtained an explicit approximation of the solution to the algebraic system~\eqref{eq:algebraic3Dsystemw}-\eqref{eq:algebraic3Dsystemv} given by~\eqref{eq:ustar(F)}. Thus, the entire system~\eqref{eq:phase linefull} is reduced to the following equation: 

\begin{equation}\label{eq:phase lineDiffbeforeapprox}
    \frac{dF}{d\tilde\tau}=\dfrac{1}{1+e^{\eta-\sigma(1+\xi F)u^*(F)}}- F,
\end{equation} where $u^*(F)$ is given by~\eqref{eq:ustar(F)}. We further simplify~\eqref{eq:phase lineDiffbeforeapprox} by using the nondimensionalized version of the approximation given in~\eqref{eq:rapproxfullminmax}. Thus,~\eqref{eq:phase lineDiffbeforeapprox} is approximated by 
\begin{equation}\label{eq:approxDF}
    \frac{dF}{d\tilde\tau}= \text{max}\left\{0,\text{min}\left\{1,\frac{1}{2}-\hat\eta+J(F)\right\}\right\}- F,
\end{equation}
where \begin{equation}
   J(F)= \hat\sigma(1+\xi F)u^*(F)=\hat\sigma(1+\xi F)\frac{a_1(1-F)}{a_2(1-F)+1},
\end{equation}
and the remaining nondimensional parameters are given in Table~\ref{tab:ndparamtableM}.

\subsubsection{Equilibrium and phase line analysis of the simplified single lake phosphorus model}

Here we discuss the possible equilibrium, their stability and bifurcation structure of equation~\eqref{eq:approxDF} with respect to the parameter $\hat\eta$. Equation~\eqref{eq:approxDF} has four possible steady state solutions given by $F^*_l=0,F^*_1=1,F^*_h,$ and $F^*_u$ where $F^*_h$ and $F^*_u$ are internal equilibrium given by the solution to $J(F)=F+\hat\eta-1/2$. The analysis is supplemented graphically in Figure~\ref{fig:phase lineinfo} where intersections of the nonlinear curve $J(F)$ with the linear curve $F+\hat\eta-1/2$ for various values of $\hat\eta$ represent the equilibria. Furthermore the structure of the equilibrium is shown in a bifurcation diagram (see Figure~\ref{fig:phase linebifurc}) where three critical values of $\hat\eta$ are highlighted. 
\begin{figure}
    \centering
    \includegraphics[width=0.6\paperwidth]{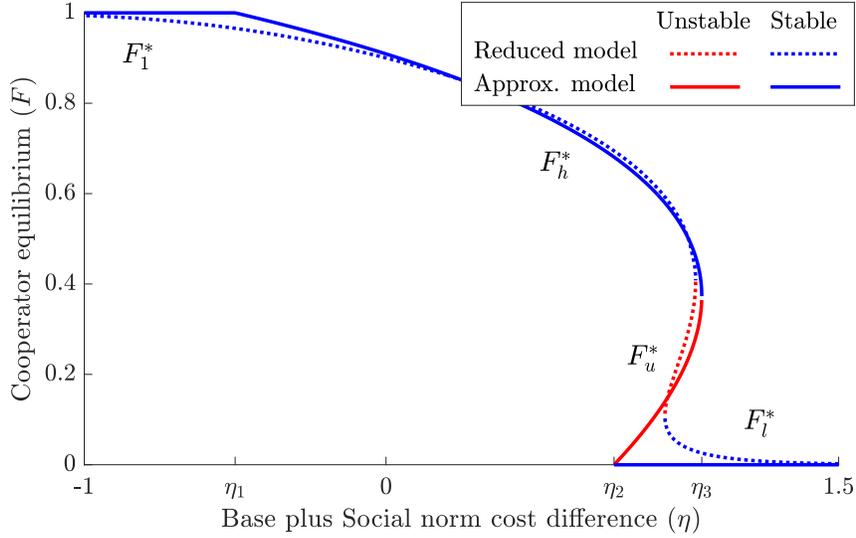}
    \caption{Solid line: Bifurcation plot of equilibria solutions to the approximated model~\eqref{eq:approxDF} with respect to $\hat\eta$. Dotted line: Bifurcation plot of the reduced model~\eqref{eq:phase lineDiffbeforeapprox} with $\eta$ values scaled to $\hat\eta$ values. Note the two plots are qualitatively similar other than $F^*_l$ is small but non zero and $F^*_1$ does not exist in the full model and is explained in Remark~\ref{rem:fulltoappphase line}.}
    \label{fig:phase linebifurc}
\end{figure}
By~\eqref{eq:approxDF} an internal equilibrium must satisfy the equation $J(F)=F+\hat\eta-1/2$ for some $F\in[0,1]$. 
 \begin{align}
     J(F)&=F+\hat\eta-1/2\label{eq:J(F)equalseta},
         \\
        \iff   0&=-F+1/2-\hat{\eta} +\hat\sigma(1+\xi F)\frac{a_1(1-F)}{a_2(1-F)+1},
  \\
  &=(1/2-\hat{\eta}-F)(a_2(1-F)+1)+\sigma(1+\xi F)(a_1(1-F)),
  \\
  &=F^2(a_2-\sigma a_1\xi)+F(\sigma a_1(\xi-1)-1-a_2-a_2(1/2-\hat{\eta}))\nonumber
  \\&\qquad \qquad\qquad\qquad\qquad \qquad\qquad+\sigma a_1+(1/2-\hat{\eta})(a_2+1)\label{eq:Fquad},
 \end{align}
 for some $F\in[0,1]$. 
 The solutions to this equation are shown graphically in Figure~\ref{fig:intersectionphase line} as the intersections of the curve $J(F)$ with $F+\hat\eta-1/2$. 
Let $\Delta$ denote the discriminant of~\eqref{eq:Fquad}. Then \begin{align}
        \Delta&=\left[\sigma a_1(\xi-1)-1-a_2-a_2(\frac{1}{2}-\hat{\eta}))\right]^2\nonumber
        \\&\qquad-4(a_2-\sigma a_1\xi)\left(\sigma a_1+(\frac{1}{2}-\hat{\eta})+a_2(\frac{1}{2}-\hat{\eta})\right),
        \\
        & =a_2^2(\frac{1}{2}-\hat{\eta})^2+\hat B(\frac{1}{2}-\hat{\eta})+\hat C, 
    \end{align}
    where $\hat B=\left[2\,a_{2}\,\left(a_{2}-a_{1}\,\sigma \,\left(\xi -1\right)+1\right)-\left(a_{2}+1\right)\,\left(4\,a_{2}-4\,a_{1}\,\sigma \,\xi \right)\right]$ and $\hat C=(a_{2}-a_{1}\,\sigma \,(\xi -1)+1)^2-a_{1}\,\sigma \,\left(4\,a_{2}-4\,a_{1}\,\sigma \,\xi \right)$.
 Two solutions to~\eqref{eq:J(F)equalseta} exist when $\Delta>0$ however, since $\Delta$ is given as a quadratic function in $1/2-\hat\eta$, $\Delta$ is not positive everywhere for all values of $\hat \eta$.  Note that for the given parameter values $B^2-4a_2^2\hat C>0$, thus $\Delta=0$ has two solutions given by

  \begin{align}
          \hat{\eta}_{3,4}
    &= -\left(\dfrac{-\hat B \pm \sqrt{\hat B^2-4a_2^2\hat C}}{2a_2^2}\right)+\frac{1}{2}\nonumber,
    \\&= \dfrac{\splitdfrac{\pm4\,\sqrt{-a_{1}\,\sigma \,\left(a_{2}-a_{1}\,\sigma \,\xi \right)\,\left(a_{2}+\xi +a_{2}\,\xi \right)}}{-2\,a_{2}-{a_{2}}^2+2\,a_{1}\,a_{2}\,\sigma +4\,a_{1}\,\sigma \,\xi +2\,a_{1}\,a_{2}\,\sigma \,\xi }}{2\,{a_{2}}^2}\label{eq:etaroots},
  \end{align} where $a_{1}\,\sigma \,\xi-a_{2}>0$ for our parameter region. Numerically we have $\hat{\eta}_3=1.0464$ and $\hat{\eta}_4=32.3$. Thus, if $\hat{\eta}<\hat{\eta}_3$ two solutions exist to~\eqref{eq:Fquad}. Also, if $\hat\eta>\eta_4$ two solutions to~\eqref{eq:Fquad} exist, but the solutions are values of $F$ that are much greater than one and are not considered. These solutions occur for values of $F$ that exceed the vertical asymptote of $J(F)$. Thus, we conclude that $\eta<\eta_3$ is a necessary condition for solutions to~\eqref{eq:J(F)equalseta} to be in $[0,1]$ and that the solutions are given by
\begin{align}
    F^*_h&= \frac{3\,a_{2}-2\,a_{2}\,\hat{\eta} +2\,a_{1}\,\sigma -2\,\sqrt{\Delta}-2\,a_{1}\,\sigma \,\xi +2}{4\,\left(a_{2}-a_{1}\,\sigma \,\xi \right)},\label{eq:Fu}
    \\ 
   F^*_u&= \frac{3\,a_{2}-2\,a_{2}\,\hat{\eta} +2\,a_{1}\,\sigma +2\,\sqrt{\Delta}-2\,a_{1}\,\sigma \,\xi +2}{4\,\left(a_{2}-a_{1}\,\sigma \,\xi \right)} \label{eq:Fl},
\end{align}
with $F^*_h>F^*_u$. Furthermore,
\begin{equation}
    \dfrac{dF^*_u}{d\hat{\eta}}=-\frac{-a_2+\frac{\frac{d\Delta}{d\hat{\eta}}}{2\sqrt{\Delta}}}{2\,\left(a_{1}\,\sigma \,\xi-a_{2} \right)}.
\end{equation} When $\hat\eta<\hat\eta_3$, $\frac{d\Delta}{d\hat{\eta}}<0$ since $\Delta$ is a concave up quadratic and $\hat\eta_3$ is the left root. Thus it is easily verified that $\dfrac{dF^*_u}{d\hat{\eta}}>0$ when $\hat\eta<\hat\eta_3$. 

Observe that

\begin{equation}
    \dfrac{dF^*_h}{d\hat{\eta}}=\frac{a_2+\frac{\frac{d\Delta}{d\hat{\eta}}}{2\sqrt{\Delta}}}{2\,\left(a_{1}\,\sigma \,\xi-a_{2}. \right)}.
\end{equation}
We show that $a_2+\frac{\frac{d\Delta}{d\hat{\eta}}}{2\sqrt{\Delta}}<0$. Since $\Delta$ is positive and $\frac{d\Delta}{d\hat{\eta}}$ is negative for $\hat{\eta}<\hat{\eta}_3$ we examine
\begin{align}
    2a_2\sqrt{\Delta}&<-\frac{d\Delta}{d\hat{\eta}},
    \\
   \iff 4a_2^2\Delta&<(-\frac{d\Delta}{d\hat{\eta}})^2,
    \\
   \iff 4a_2^2(a_2^2\hat{\eta}^2+\hat B\hat{\eta}+\hat C)&<(2a_2^2\hat{\eta}+\hat B)^2,
    \\
  \iff  4a_2^4\hat{\eta}^2+4a_2^2\hat B\hat{\eta}+4a_2^2\hat C&<4a_2^2\hat{\eta}^2+2a_2^2\hat B\hat{\eta}+\hat B^2,
    \\
  \iff  0<\hat B^2-4a_2^2\hat C,
\end{align}
which is verified true as in~\eqref{eq:etaroots}. Thus, $\dfrac{dF^*_h}{d\hat{\eta}}<0$ for $\hat\eta<\hat\eta_3$.

Since the internal equilibrium of our system are given by $F^*_h$ and $F^*_u$, and we have further shown that $F^*_u$ is an increasing function of $\hat\eta$, while $F^*_h$ is decreasing. We now search for the critical points for $\hat\eta$ in which $F^*_h<1$ and $F^*_u>0$. First, when $\hat\eta=\hat\eta_1=-1/2$, we note that $F^*_h=1$ as seen visually from Figure~\ref{fig:intersectionphase line} and  verified mathematically from both~\eqref{eq:Fu} and~\eqref{eq:J(F)equalseta}. Second, when  $\hat\eta=\hat\eta_2=1/2+\sigma a_1/(a_2+1)$,
 we note that $F^*_u=F^*_l=0$ as seen visually from Figure~\ref{fig:intersectionphase line} and verified mathematically from both~\eqref{eq:Fl} and~\eqref{eq:J(F)equalseta}.
Lastly, we note that $F^*_h=F^*_u$ when $\hat\eta=\hat\eta_3$ and for $\hat\eta>\hat\eta_3$, $F^*_h$ and $F^*_u$ have imaginary parts and are not considered equilibrium. 
The above discussion leads us to the following theorem regarding stability and bifurcations of~\eqref{eq:approxDF}.
\begin{theorem}\label{thm:phase lineEQ}The equilibrium and stability of~\eqref{eq:approxDF} are given by the following:
\begin{itemize}
    \item (i) If $\hat\eta<\hat\eta_1=-1/2$ then $F^*_1$ is the only equilibrium to exist and is globally stable. 
    \item (ii) If $\hat\eta_1<\hat\eta<\hat\eta_2$ then $F^*_h$ is the only equilibrium to exist and is globally stable. 
    \item (iii) If $\hat\eta_2<\hat\eta<\hat\eta_3$ then $F^*_l,F^*_u$ and $F^*_h$ exist. Bistability occurs where $F^*_l$ and $F^*_h$ are locally stable, and $F^*_u$ is unstable. 
    \item (iv) If $\hat\eta>\hat\eta_3$ then $F^*_l$ is the only equilibrium to exist and is globally stable.
\end{itemize}
\end{theorem}
The proof for Theorem \ref{thm:phase lineEQ} is given in the \nameref{Appendix}.
Furthermore, we define a saddle node bifurcation as a point when two steady states collide and annihilate each other as a bifurcation parameter changes. Following this definition we conclude the following corollaries.
\begin{corollary}\label{cor:saddlenode}
A saddle node bifurcation occurs at the point $(F,\hat\eta)=(F^*_h,\hat\eta_3)$. 
\end{corollary}
\begin{corollary}\label{cor:saddlenodezero}
A saddle node bifurcation occurs at the point $(F,\hat\eta)=(F^*_l,\hat\eta_2)$. 
\end{corollary}
The bifurcations discussed in \ref{cor:saddlenode} and \ref{cor:saddlenodezero} are illustrated intuitively in Figures~\ref{fig:phase linebifurc} and by noting the intersections of $J(F)$ with $F+\hat\eta-1\/2$ in Figure~\ref{fig:phase lineinfo}. The detailed proofs are given in the \nameref{Appendix}.

\begin{remark}\label{rem:fulltoappphase line}
In the non-approximated model~\eqref{eq:phase lineDiffbeforeapprox} the equilibrium  $F^*_1=1$ does not exist. The equilibrium $F^*_h$ approaches one, but does not equal and there is no transition from $F^*_1$ to $F^*_h$ seen in Figure~\ref{fig:phase linebifurc}. Furthermore, the equilibrium $F^*_l$ will be small, but non zero. These differences are explained by the linear approximation made at the tails of the logistic curve being set to one or zero accordingly as shown in Figure~\ref{fig:pwiseapp}. Furthermore, we conjecture that the flows near corresponding equilibrium are topologically equivalent between the two models~\eqref{eq:phase lineDiffbeforeapprox} and~\eqref{eq:approxDF} and that the analysis presented for~\eqref{eq:phase lineDiffbeforeapprox} holds for~\eqref{eq:approxDF}. 
\end{remark}
Theorem~\ref{thm:phase lineEQ} gives an understanding of the possible regime outcomes of the single lake dynamics based on the parameter $\hat\eta$, which describes the difference in both baseline and external social norm costs between the two strategies. The results of Theorem~\ref{thm:phase lineEQ} are summarized graphically in Figure~\ref{fig:phase lineinfo}.

\begin{figure}
\centering
\begin{subfigure}[b]{0.45\textwidth}
     \centering
    \includegraphics[width=\textwidth]{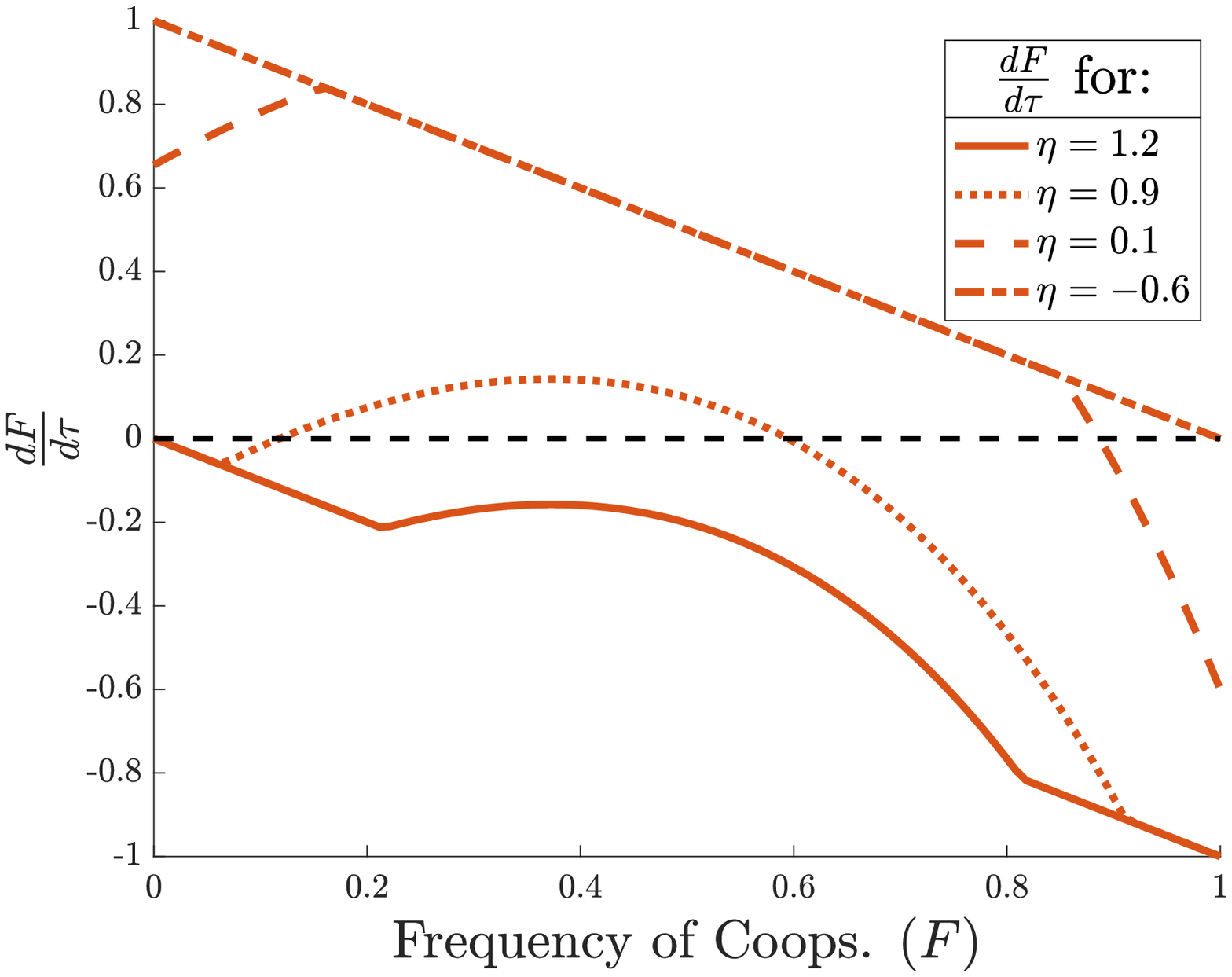}
    \caption{}
    \label{fig:phase line}
\end{subfigure}
\hfill
\begin{subfigure}[b]{0.45\textwidth}
      \centering
    \includegraphics[width=\textwidth]{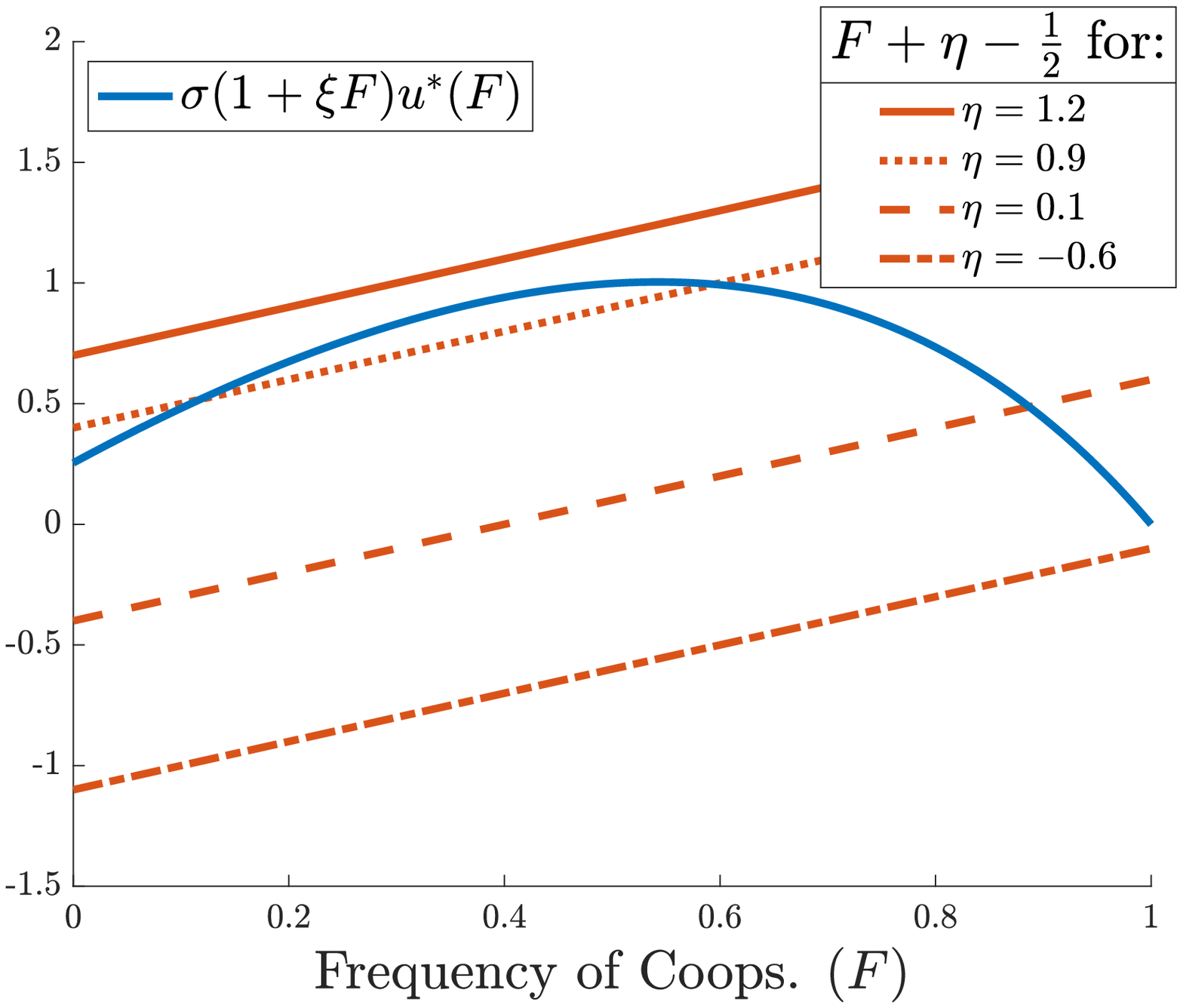}
    \caption{}
    \label{fig:intersectionphase line}
\end{subfigure}
\caption{(a) Phase line of equation~\eqref{eq:approxDF} for the four cases given in Theorem~\ref{thm:phase lineEQ}. (b)  Curve $J(F)$ and line $F+\hat{\eta}-1/2$ for four values of $\hat\eta$ corresponding to the cases in Theorem~\ref{thm:phase lineEQ}. The points of intersection give the equilibria of equation~\eqref{eq:approxDF} }
\label{fig:phase lineinfo}
\end{figure}

\subsection{Dynamics of the iron explicit model}\label{sec:phase plane}

Here we consider different parameter values where the nutrient of focus is iron instead of phosphorus. These new parameter values allow us to look at the dynamics of the coupled CB and socio-economic model in the phase plane. 
The parameter values chosen for Section~\ref{sec:phase line} (given in Tables~\ref{tab:ecoparamtable} and~\ref{tab:socecparamtable}) represent a typical system in which phosphorus pollution occurs. However, we now consider the situation where iron is the focal nutrient. When considering iron instead of phosphorus we must alter certain assumptions and parameters in our model. First, the values of $Q_M$ and $Q_m$ are decreased by nearly an order of magnitude as implied by the extended Redfield ratio~\cite{Cunningham2017,North2007}. Similarly, the uptake rate ($\rho_m$) is smaller~\cite{Downs2008,Larson2015,Cunningham2017}. However the half saturation constant may not need to decrease, meaning the phytoplankton are inefficient at `finding' iron at low concentrations. As before, we nondimensionalize the system, but where $P$ now represents the iron concentration and the values of only the following parameters are changed: $Q_M=.4e^{-4}$,$Q_m=.4e^{-5}$, $\rho_M=1e^{-3}$, $\xi=5$, $p_\defe=100$,$p_\coop=50$. All other parameter values remain as in Tables~\ref{tab:ecoparamtable} and~\ref{tab:socecparamtable} but are interpreted for iron.

\subsubsection{Nondimensionalization of the single lake model}
We now continue with the non-dimensionalization of system~\eqref{eq:nondimsinglelake} by the making the substitutions, $\tau=rt$, $u=kz_mB$, $v=\frac{Q}{Q_M}$, $w=\frac{P}{p_\coop}$, yielding:
\begin{equation}\label{eq:nondim2}
    \begin{dcases}\left.
       \begin{aligned}
       \frac{du}{d\tau}&=u(1-\frac{Q_m}{Q_M}\frac{1}{v})h(au)-\frac{(\nu_r+\frac{D}{z_e})}{r} u,
       \\
 \frac{dv}{d\tau}&=\frac{\rho_m}{rQ_M}\frac{Q_M-Q_Mv}{Q_M-Q_m}\frac{w}{M/p_\coop+w}-(v-\frac{Q_m}{Q_M})h(au),
\\
p_\coop\frac{dw}{d\tau}&=\frac{D}{rz_e}(p_\coop F+p_\defe(1-F)-p_\coop w)-\frac{\rho_M}{rkz_e}u\frac{Q_M-Q_M v}{Q_M-Q_m}\frac{w}{M/p_\coop+w},
\\
\frac{dF}{d\tau}&=\frac{s}{r}\left(\dfrac{1}{1+e^{\beta(c_C-c_\mathcal{D}-\alpha(1+\xi F)\psi au +\hat\delta)}}-F\right).
       \end{aligned} \right. 
    \end{dcases}
\end{equation}
Upon substitution of the dimensionless parameters that given in Table~\ref{tab:ndparamtablephase plane} we arrive at:
\begin{equation}\label{eq:nondimPF}
    \begin{dcases}\left.
       \begin{aligned}
       \frac{du}{d\tau}&=u(1-\frac{1}{\gamma}\frac{1}{v})\hat{h}(u)-(\epsilon\beta_1+\beta_2) u,
       \\
 \frac{dv}{d\tau}&=\omega(1-v)\frac{w}{\mu+w}-(v-\frac{1}{\gamma})\hat{h}(u),
\\
\frac{dw}{d\tau}&=\epsilon\left(\beta_1( F+ \kappa(1-F)-w)-\lambda u(1-v)\frac{w}{\mu+w}\right),
\\
\frac{dF}{d\tau}&=\epsilon\left(\dfrac{1}{1+e^{\hat{\eta}-\sigma(1+\xi F)u)}}-F\right),
       \end{aligned} \right. 
    \end{dcases}
\end{equation}
where $\hat{h}(u)$ is the nondimensional light dependent growth originating from~\eqref{eq:h(B)} given by~\eqref{eq:h(u)nondim} and approximated by~\eqref{eq:hhatapp}.

\begin{table}[h!]
\caption{Dimensionless parameters for the iron system~\eqref{eq:nondimPF}}
\centering
{
\begin{tabular}{|l |c| c |} 
\hline
Parameter & Definition & Value
\\
\hline
$\beta_1$&$\dfrac{D}{s z_m}$ & 0.2857 
\\
$\beta_2$&$\nu_r/r$& 0.35 
\\
 $\omega$&$\dfrac{\rho_m}{r(Q_M-Q_m)}$& 2.7778 
 \\
 $\gamma$& $\frac{Q_M}{Q_m}$ & 10
\\
 $\kappa$& $\frac{p_\mathcal{D}}{p_\coop}$ & 2
 \\
 $\mu$ & $\frac{M}{p_\coop}$ & 0.03
\\
$\lambda$&$\dfrac{Q_M}{Q_M-Q_m}\dfrac{\rho_m }{p_\coop    s k z_m}$&0.7937
\\
$k_1$&$z_m K_{b g}$& 2.1
\\
$I$&$I_{in}/H$ & 2.5
\\
$\eta$&$\beta(c_\mathcal{C}-c_\mathcal{D}+\hat\delta)$& - 
\\
$\sigma$&$\alpha\beta \psi / kz_m$& 2.1429
\\
$\epsilon$&$s/r$& $<$0.001
\\\hline
\end{tabular}}
\label{tab:ndparamtablephase plane}
\end{table}

\subsubsection{Application of the quasi steady state approximation}
We now apply the QSSA to~\eqref{eq:nondimPF}. As in Section~\ref{sec:phase line} we introduce the new timescale $\tilde\tau=\epsilon\tau$. $\tilde \tau$ now represents the slow timescale in which the socio-economic and iron dynamics mainly occur, whereas the CB growth dynamics occur on the fast time scale $\tau$. Lastly, we apply the QSSA and let $\epsilon\to 0$ arriving at the system:

  \begin{subequations}\label{eq:phase planefull}
        \begin{empheq}[left=\empheqlbrace]{align}
           0&=u(1-\frac{1}{\gamma}\frac{1}{v})\hat{h}(u)-\beta_2 u\label{eq:pplaneu},
       \\
 0&=\omega(1-v)\frac{w}{\mu+w}-(v-\frac{1}{\gamma})\hat{h}(u)\label{eq:pplanev},
\\
\frac{dw}{d\tilde\tau}&=\beta_1( F+\kappa(1-F)-w)-\lambda u(1-v)\frac{w}{\mu+w}=g(F,w)\label{eq:pplaneg},
\\
\frac{dF}{d\tau}&=\dfrac{1}{1+e^{\eta-\sigma(1+\xi F)u)}}-F=f(F,w)\label{eq:pplanef},
  \end{empheq}
  \end{subequations}
reducing our problem to a differential- algebraic system.  Denote $u^*(w)$ and $v^*(w)$ as a solution to the algebraic system defined by~\eqref{eq:pplaneu} and~\eqref{eq:pplanev}. The following theorem gives a condition to guarantee existence of uniqueness of a solution to the algebraic system.

\begin{theorem} There exists a unique positive solution, $(\hat u(w),\hat v(w))$, to the algebraic system defined by~\eqref{eq:pplaneu} and~\eqref{eq:pplanev} if $\left(1-\dfrac{1}{\gamma \hat{v}(w)}\right)\hat h(0)-\beta_2>0$.  The trivial solution $(0,\bar v(w))$ always exists.
 \end{theorem}
 \begin{proof}
Observe that $(0,\bar v(w))$ where \begin{equation}\label{eq:vbar}
     \bar v(w)= \dfrac{\omega\dfrac{w}{\mu+w}+\frac{1}{\gamma}\hat{h}(0)}{\omega\dfrac{w}{\mu+w}+\hat{h}(0)},
 \end{equation} solves~\eqref{eq:pplaneu} and~\eqref{eq:pplanev}. Now, we compute the positive solution by first adding equation~\eqref{eq:pplaneu} multiplied by $v/u$ to equation~\eqref{eq:pplanev} arriving at
 \begin{equation}
    0= \omega(1-v)\frac{w}{\mu+w}-\beta_2v.
 \end{equation}
Thus,
\begin{equation}\label{eq:vstar}
  \hat{v}(w)=\dfrac{\omega\dfrac{w}{\mu+w}}{\omega\dfrac{w}{\mu+w}+\beta_2}=\dfrac{\omega w}{\omega w+\beta_2(\mu+w)}.  
\end{equation}
 Now we see that $\hat v(w)$ is only dependent on parameters as $w$ is treated as a parameter in the algebraic system.  Thus,~\eqref{eq:pplaneu} can be reduced to a problem with a single unknown: 
\begin{equation}\label{eq:solve for u eq}
    0=(1-\frac{1}{\gamma \hat{v}(w)})\hat h(u)-\beta_2.
\end{equation}

 Recall, by the construction of $h(B)$ in~\eqref{eq:h(B)} it and its nondimensional analog, $\hat{h}(u)=\frac{1}{u+k_1}\text{log}\left(\frac{1+I}{1+I\textup{exp}(-u-k_1)}\right)$, are monotone decreasing functions and $\lim_{u\to \infty}\hat{h}(u)=0.$ Thus, if
 \begin{equation}\label{eq:2dthmrestriction}
     (1-\frac{1}{\gamma \hat{v}(w)})\frac{1}{k}\text{log}(\frac{1+I}{1+Ie^{-k}})-\beta_2>0,
 \end{equation} then a unique positive solution exists via the intermediate value theorem.
  \end{proof}

Note that the condition in the theorem is also satisfied for values of $w$, such that $w>w_c$, where $w_c$ is the critical point such that $\bar v^(w_c)=\hat{v}(w_c)$. $w_c$ can be written explicitly as \begin{equation}
    w_c=-\dfrac{\frac{1}{\gamma}\hat h(0)\mu\beta_2}{\omega\beta_2+\frac{1}{\gamma}\hat h(0)\omega+\frac{1}{\gamma}\beta_2\hat h(0)-\omega\hat h(0)}.
\end{equation} The equation $\bar v(w)=\hat{v}(w)$ is reduced to a linear equation in $w$ thus verifying $w_c$ is unique. 
  For $w<w_c$ the only solution to the system~\eqref{eq:pplaneu} and~\eqref{eq:pplanev} is given by  $(0,\bar v(w))$.  For $w>w_c$ the positive solution $(\hat{u}(w),\hat{v}(w))$ also exists and $(0,\bar v(w_c)=(\hat{u}(w_c),\hat{v}(w_c))$. It can be shown that the trivial equilibrium of the fast subsystem is unstable if $w>w_c$ thus, we take the solutions to system~\eqref{eq:pplaneu} and~\eqref{eq:pplanev} as
\begin{equation}\label{eq:ustarthm}
    u^*(w)=\text{max}\{0,\hat u(w)\}\textup{  \qquad and  \qquad }v^*(w)=\textup{max}\{\bar{v}(w),\hat v(w)\},
\end{equation}
where $\hat{u}(w)$ is the solution to~\eqref{eq:solve for u eq}.

 \subsubsection{Phase Plane analysis of the simplified single lake iron model}\label{sec:singlelakephase planeanalysis}
 We now proceed with studying~\eqref{eq:phase planefull} in the phase plane. Changing the value for $\eta$ in system~\eqref{eq:phase planefull} will result in various phase portraits that are topologically different. One such instance shows a bistable scenario that is lost through either one of two saddle-node bifurcations.
 
 In the following we assume that $u^*(w),v^*(w)$ are defined as in~\eqref{eq:ustarthm}. The bifurcation plot, with respect to $\eta$ of system~\eqref{eq:phase planefull} is shown in Figure~\ref{fig:pplanebifurc}. Recall that the parameter  $\eta=c_\coop-c_\defe+\hat\delta$ partly describes the differences in base costs and network social norm costs faced by the cooperator and defector, respectively. 
 \begin{figure}
    \centering
    \includegraphics[width=0.3\paperwidth]{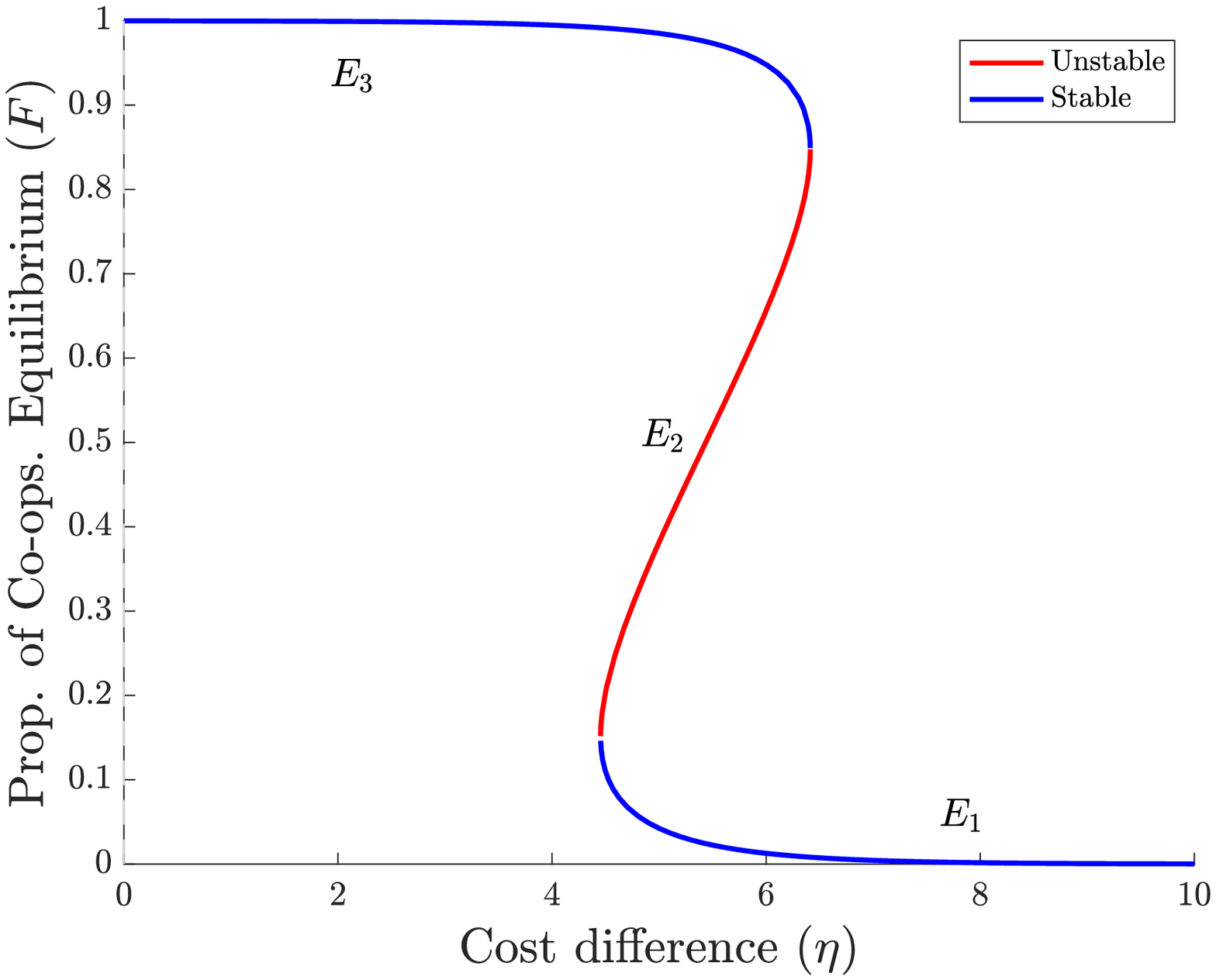}
    \includegraphics[width=0.3\paperwidth]{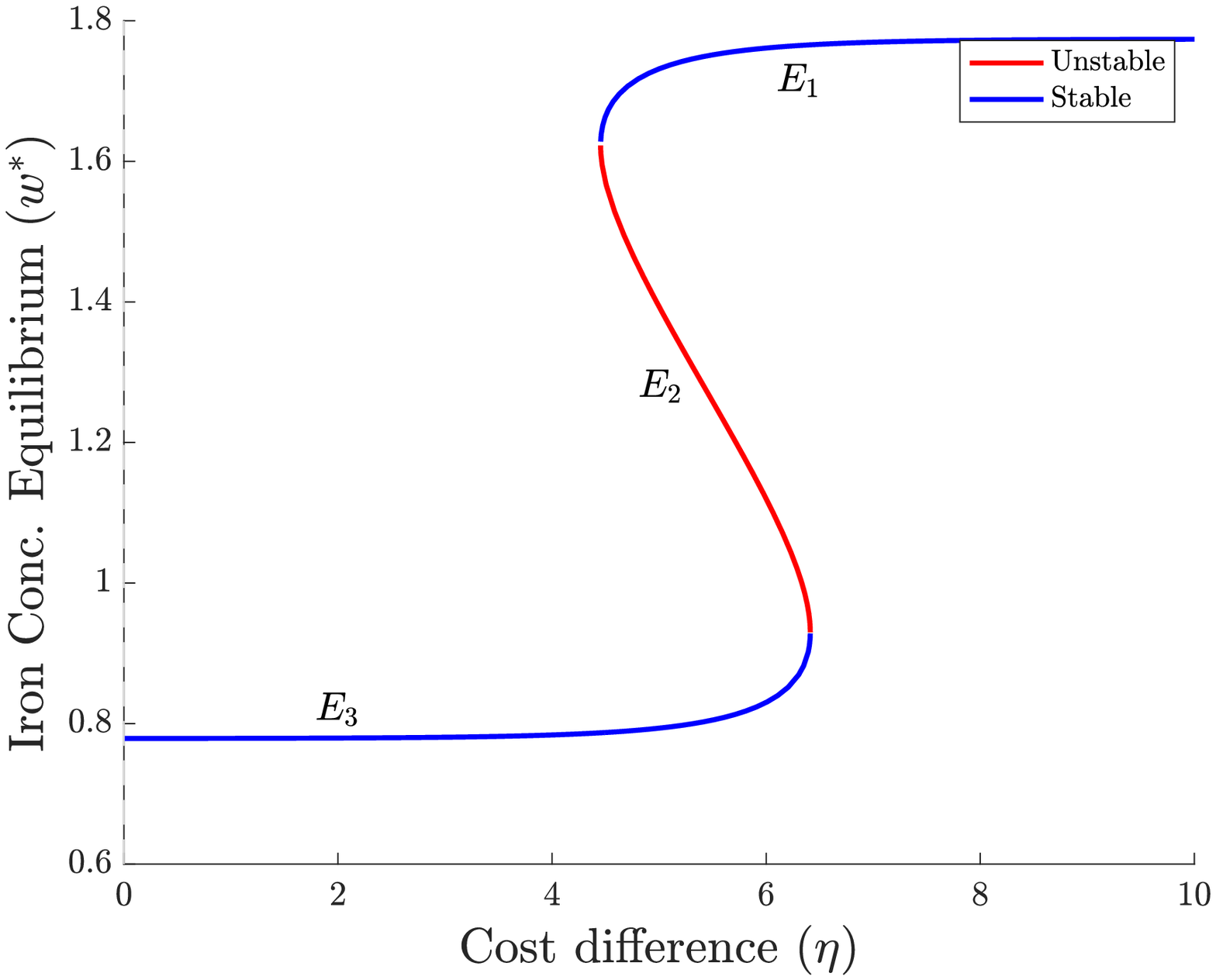}
    \caption{Bifurcation diagrams for system~\eqref{eq:pplaneg} and~\eqref{eq:pplanef}. Left: equilibrium values for the proportions of cooperators ($F$). Right: Equilibrium values for the concentration of iron.}
    \label{fig:pplanebifurc}
\end{figure}

 The first case we explore is for $\eta$ values that are relatively `large'. 
 \begin{figure}
     \centering
     \includegraphics[width=0.5\paperwidth]{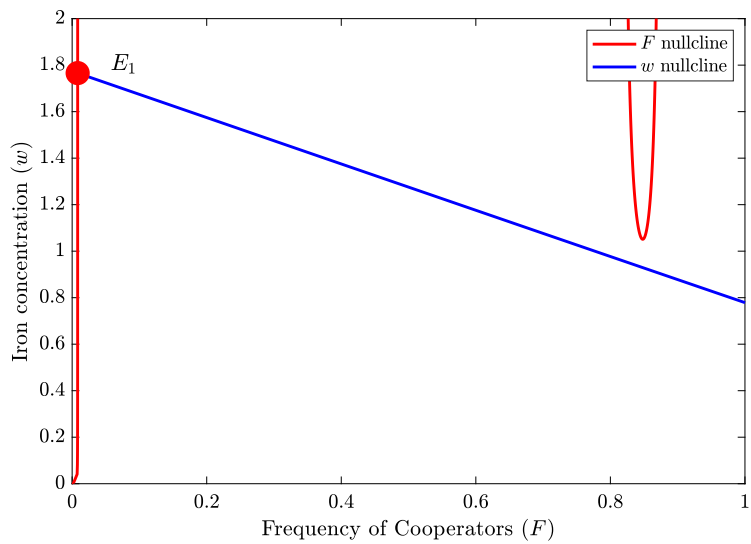}
     \caption{Phase plane for `large' values of $\eta$. $E_1$ is the only equilibrium and attracts all solutions.}
     \label{fig:casee1}
 \end{figure}
 Figure~\ref{fig:casee1} shows the phase portrait for values of $\eta$ that represent the base costs of cooperating to be relatively much higher than that of the defector. In this case the equilibrium $E_1$, defined by a low level of cooperation and a relatively high concentration of iron, is numerically globally attracting and the only equilibrium to exist. We show via graphical arguments the local stability below:

 It is enough to observe the sign changes of $f(F,w)$ and $g(F,w)$ as we cross the nullclines near the equilibrium point $E_1$  given in the phase plane (see Figure~\ref{fig:casee1}). 
 The Jacobian of system (~\eqref{eq:pplaneg} and~\eqref{eq:pplanef}) has the following form: \begin{equation}
 A|_{E_1}=  \begin{pmatrix}
f_F & f_w \\
g_F & g_w 
\end{pmatrix}_{E_1}=
 \begin{pmatrix}
 - & + \\
 - & - \end{pmatrix}.
\end{equation} 
By the signs given in $A|_{E_1}$ we conclude that,  $Tr(A|_{E_1})<0$ and $det(A_{E_1})>0$. Thus, the matrix $A|_{E_1}$ has eigenvalues with negative real parts and we conclude that equilibrium $E_1$ is locally stable.


 Now we look at the case where $\eta$ is relatively `small'. Here the values of $\eta$ represent the base costs of cooperating to be relatively close to that of the defectors. 
 \begin{figure}
     \centering
     \includegraphics[width=0.5\paperwidth]{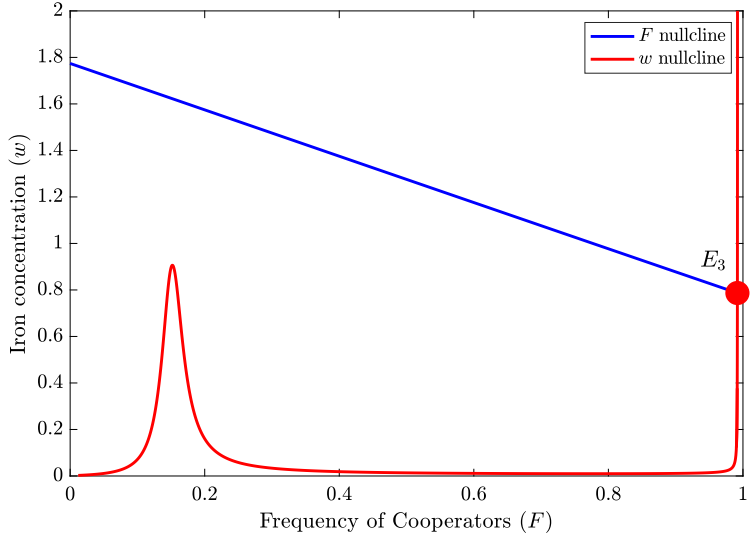}
     \caption{
     Phase plane for `small' values of $\eta$. $E_3$ is the only equilibrium and attracts all solutions.
     }
     \label{fig:casee3}
 \end{figure}
 Figure~\ref{fig:casee3} shows the phase portrait for this scenario. Here, the equilibrium $E_3$, defined by a high level of cooperation and a relatively low concentration of iron, is numerically globally attracting and the only equilibrium to exist. This means the phase portrait in Figure~\ref{fig:casee3} represents an environmentally favourable outcome. We show the local stability below.

Near the equilibrium point, $E_3$ in Figure~\ref{fig:casee3} , we observe the sign changes of $f(F,w)$ and $g(w,F)$ as the nullclines are crossed. The Jacobian of system~\eqref{eq:pplaneg} and~\eqref{eq:pplanef} has the following form: \begin{equation}
 A|_{E_3}=  \begin{pmatrix}
f_F & f_w \\
g_F & g_w 
\end{pmatrix}_{E_3}=
 \begin{pmatrix}
 - & + \\
 - & - \end{pmatrix}.
\end{equation} 
We conclude that $Tr(A|_{E_3})<0$ and $det(A_{E_3})>0$. Thus, the matrix $A|_{E_3}$ has eigenvalues with negative real parts and we conclude that equilibrium $E_3$ is locally stable.

 Now we show the dynamics for intermediate values of $\eta$ such that a bistable scenario occurs. These values of $\eta$ represent an intermediate region where the cost of cooperating and defecting are close to being balanced. 
 \begin{figure}
     \centering
     \includegraphics[width=0.5\paperwidth]{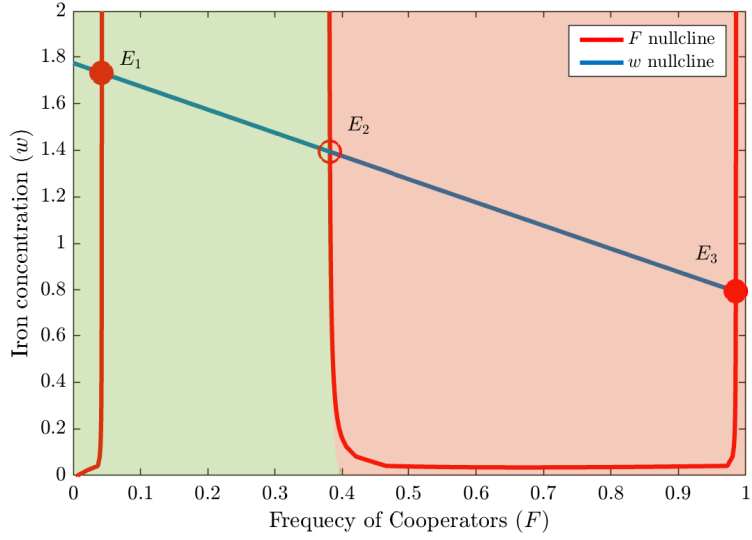}
     \caption{Phase plane for `intermediate' values of $\eta$. In this case $E_1$ and $E_3$ have similar topology as in Figures~\ref{fig:casee1} and~\ref{fig:casee3} and thus have the same stability. The shaded regions represent the attraction basin of the respective equilibrium. $E_2$ is unstable. 
     }
     \label{fig:casee1e2e3}
 \end{figure}
Figure~\ref{fig:casee1e2e3} shows the phase portrait of the bistable scenario, where the topology near equilibrium $E_1$ and $E_3$ are qualitatively consistent to that shown in Figures \ref{fig:casee1} and \ref{fig:casee3}, respectively. That is, $E_1$ and $E_3$ are both locally stable for our chosen parameter region. A new equilibrium, $E_2$, appears as seen in Figure~\ref{fig:casee1e2e3} which we now show is unstable. First by graphical methods, the Jacobian evaluated near $E_2$ is, 
 \begin{equation}
 A|_{E_2}=  \begin{pmatrix}
f_F & f_w \\
g_F & g_w 
\end{pmatrix}_{E_2}=
 \begin{pmatrix}
 + & + \\
 - & - \end{pmatrix}.
\end{equation} 
From this alone, we can not make any conclusions.  Graphically, the gradients evaluated on the nullclines  are such that $\frac{dw}{dF}|_{g=0}<0$, $\frac{dw}{dF}|_{f=0}<0$, and  $\frac{dw}{dF}|_{g=0}>\frac{dw}{dF}|_{f=0}$, near $E_2$. Furthermore,  \begin{equation}\left.
    \dfrac{dw}{dF}\right|_{g=0}=-\dfrac{g_F}{g_w}>  \left.\dfrac{dw}{dF}\right|_{f=0}=-\dfrac{f_F}{f_w},
\end{equation}
\begin{align*}
    \implies& -\dfrac{g_F}{g_w}>-\dfrac{f_F}{f_w},
    \\
     \implies& \dfrac{g_F}{g_w}<\dfrac{f_F}{f_w},
     \\
     \implies& f_w\dfrac{g_F}{g_w}<f_F,
     \\
     \implies& f_w g_F>f_F g_w,
     \\
     \implies& 0>f_F g_w-f_w g_F,
     \\
     \implies& det(A)<0.
\end{align*}
Thus, $E_2$ is an unstable saddle. Note that the above comes from  taking the derivative of the level set $f(F,w)=0$ and graphical observation. \begin{remark}
In Figures~\ref{fig:casee1} and~\ref{fig:casee1e2e3} we note that there is a region of $F$ values where the $F$ nullcline does not exist. This atypical phenomenon is explained by the dependency of the nullcline on the solution to the algebraic system~\eqref{eq:pplaneu} and~\eqref{eq:pplanev}, $u^*(w)$. That is, the $F$ nullcline is given by the solutions to \begin{align}
     0&= \dfrac{1}{1+e^{\eta-\sigma(1+\xi F)u^*(w))}}-F,
   \\
   \iff u^*(w)&=\frac{\textup{log}(\frac{1}{F}-1)-\eta}{-\sigma(1+\xi F)}.\label{eq:ustarnullcline}
\end{align}
  The right hand side of~\eqref{eq:ustarnullcline} has one local max, one local min, and one inflection point in $[0,1]$. Also, note that as $F\to 0 $ the RHS of~\eqref{eq:ustarnullcline} goes to $-\infty$ and as $F\to 1$ the RHS of~\eqref{eq:ustarnullcline} goes to $+\infty$. However, $u^*(w)$ is a saturating function of $w$, and is bounded between zero and some positive constant. The positive bound results from CB self-shading and light limitation. Thus, for certain values of $F$, the right hand side of~\eqref{eq:ustarnullcline} exceeds the upper bound of $u^*(w)$ resulting in no solution to the equation. Furthermore, vertical asymptotes occur as $F$ approaches the regions of nonexistence accordingly. 
\end{remark}

We have shown that the socio-economic and ecological regime is highly connected to costs associated with each strategy. The dynamic properties have been explored in the phase plane which included a high cooperation regime, a low cooperation regime, and a bistable scenario. 

\section{Dynamics of a network system}\label{sec:Networkmodel}
We now return to the network model proposed in~\eqref{eq:CBnetwork} with parameter values given for the phosphorus dynamics in Tables~\ref{tab:ecoparamtable} and~\ref{tab:socecparamtable}.
Here we consider the nondimensional version of~\eqref{eq:CBnetwork} and, upon further simplifying assumptions, reduce the entire network model to a system of two ODEs which are studied in the phase plane. A series of two parameter bifurcation plots are provided to show regime outcomes for various parameter regions associated with the socio-economic dynamics. 

\subsection{Reduction of the network model}
We now make a series of simplifying assumptions to reduce the network model ~\eqref{eq:CBnetwork} to a system of two ODEs. 

We wish to further understand how socio-economic pressures can lead to regime shifts. In Section~\ref{sec:singlelake} we consider these pressures on a local scale, we now consider these pressures when they originate from distant social connections. For this reason, let us consider a network with $N$ lakes. Assume that each lake is modelled with identical parameter values and that if the network is not connected, each lake exhibits the bistable dynamics discussed in the phase line analysis (Section~\ref{sec:phase line}).  Note that the nondimensionalization of the network model~\eqref{eq:CBnetwork} is equivalent to what is shown in~\eqref{eq:nondimbetaeqepsilon} with the explicit subscripts for lake $i$ and the term $\hat\delta$ is to be a function of the weighted average of the frequency of cooperators in the network. 

 
 Assume a given lake is initially either in the high cooperation state ($F^*_h$) or the low cooperation state ($F^*_l$) at equilibrium as discussed in Figures~\ref{fig:phase line} and~\ref{fig:phase linebifurc}. Since the system is bistable, we can assume that our network has $N-k$ lakes in the low cooperation state ($F^*_l$) and $K$ lakes in the high cooperation ($F^*_h$) state.

 By assuming that every lake in the network has the same environmental parameters we can conclude that the dynamics of lakes with the same initial conditions will be identical. We introduce two new state variables $F_h$ and $F_l$. $F_h$ represents the dynamics of frequency of cooperators that start in a high cooperation state. $F_l$ is the dynamics of those that start in a low cooperation state. Indeed a requirement based on these assumptions is that $F_h(0)>F_l(0)$.
 

Then, $\hat\delta(\bar F(t))$ is rewritten as 
 \begin{equation}\label{eq:deltahat}
     \hat\delta(\bar F(t))=\hat\delta(F_l,F_h)=\delta_\coop(1-\bar F)-\delta_\defe \bar F.
 \end{equation} Since all lakes of the same regime are in the same state, then $\bar F=\frac{(k)F_{h}+(N-k)F_{l}}{N}$. Furthermore, assuming that the ecological dynamics of each lake are in steady state and occur on a faster timescale (by the QSSA discussed in Section~\ref{sec:phase lineQSSA}) the entire network is reduced to a two dimensional system of equations: 
 \begin{subequations}\label{eq:2Dnetwork}
       \begin{align}
         \frac{dF_h}{d\tau_1}&=\frac{1}{1+e^{\eta-u^*(F_h)(1+\xi F_h)+\delta(F_l,F_h)}}-F_h=f_1(F_h,F_l),
         \\
       \frac{dF_l}{d\tau_1}&=\frac{1}{1+e^{\eta-u^*(F_l)(1+\xi F_l)+\delta(F_l,F_h)}}-F_l=f_2(F_h,F_l),
       \end{align}
\end{subequations}
where $u^*(F)$ is the unique equilibrium of the ecological system discussed in Theorem~\ref{thm:exunphase line}. The system~\eqref{eq:2Dnetwork} then represents the coupled socio-economic and ecological dynamics of a network of lakes where $F_h$ and $F_l$ represent the social dynamics of lakes that start in a high and low cooperation regime, respectively.

\subsection{Phase plane analysis of the reduced network model}

We now use the phase plane to explore the possible long term dynamics of the network dependent on socio-economic parameters. Three stable steady states can occur that correspond to high, low and mixed levels of cooperation throughout the network. By manipulating parameters a network shift from mixed to either high or low cooperation can occur.  

  Note that under the condition $\hat\delta(F_l,F_h)=0$ the system is decoupled and each lake would exhibit bistable behaviour. Further note that given our assumptions we must limit the phase plane to the region $F_l\leq F_h$. We assume that the system will start the prescribed initial condition of $N-K$ lake in a low cooperation state (corresponding $F_l^*$) and $k$ lakes in the high cooperation state (corresponding to $F_h^*$). This means that the initial condition of our system is $F_h(0)= F_h^*$ and $F_l(0)= F_l^*$.

\begin{figure}
    \centering
    \includegraphics[width=0.5\paperwidth]{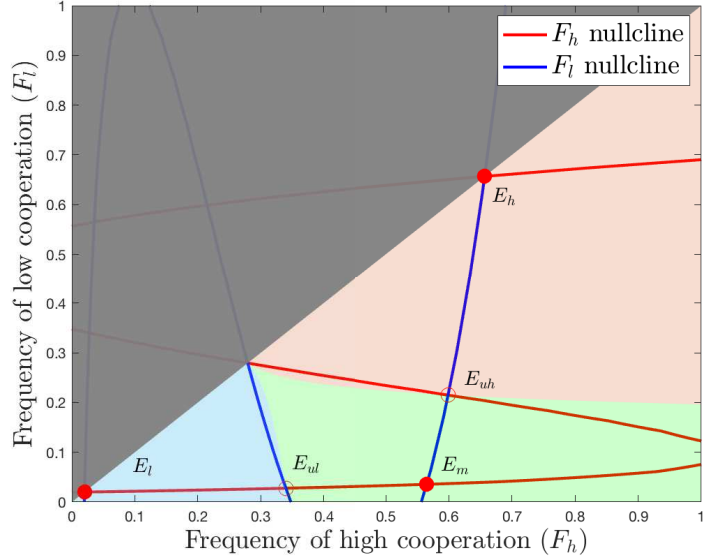}
    \caption{The phase plane of $F_h$ and $F_l$ for similar social norm pressure. The region above the line $F_l=F_h$ is excluded based on the condition $F_l<F_h$. The shaded regions represent the attraction basins of the stable equilibrium, the initial condition is located within the green region. In this case no regime shifts occur based on the prescribed initial condition. }
    \label{fig:ppEm}
\end{figure}

The stable equilibria are $E_l,E_h$, and $E_m$ and the unstable equilibria are $E_{ul}$ and $E_{uh}$ as shown in Figure~\ref{fig:ppEm}. The stability of $E_l,E_h$, and $E_m$ can be easily shown through graphical arguments similar to that in Section~\ref{sec:singlelakephase planeanalysis}. That is, the Jacobian of system~\eqref{eq:2Dnetwork} near $E_l$ has the form 
 \begin{equation}
 A|_{E_l}=  \left.\begin{pmatrix}
 \dfrac{df_1}{dF_h} & \dfrac{df_1}{dF_l} \\
 \dfrac{df_2}{dF_h} & \dfrac{df_2}{dF_l} \end{pmatrix}\right|_{E_l}
= \begin{pmatrix}
 - & + \\
 + & - \end{pmatrix}.
\end{equation} 
Thus, $Tr( A|_{E_l})<0$. Moreover,
\begin{equation}\left.
    \dfrac{dF_l}{dF_h}\right|_{f_1=0}=-\dfrac{df_1}{dF_h}/\dfrac{df_1}{dF_l}>  \left.\dfrac{dF_l}{dF_h}\right|_{f_2=0}=-\dfrac{df_2}{dF_h}/\dfrac{df_2}{dF_l},
\end{equation}
\begin{align*}
    \implies& -\dfrac{df_1}{dF_h}/\dfrac{df_1}{dF_l}>-\dfrac{df_2}{dF_h}/\dfrac{df_2}{dF_l},
    \\
    \implies & -\dfrac{df_1}{dF_h}\dfrac{df_2}{dF_l}<-\dfrac{df_1}{dF_l}\dfrac{df_2}{dF_h},
    \\
    \implies & det(A|_{E_l})>0. 
\end{align*}
Thus, $E_l$ is stable. The stability of $E_m$ and $E_h$ can be verified in an identical fashion as $E_l$. The instability of $E_{uh}$ and $E_{ul}$ is concluded similarly with 

\begin{equation}
 A|_{E_{ul}}= \begin{pmatrix}
 + & + \\
 + & - \end{pmatrix} \implies det( A|_{E_{ul}})<0,
\end{equation} 
and 
\begin{equation}
 A|_{E_{uh}}= \begin{pmatrix}
 - & + \\
 + & + \end{pmatrix} \implies det( A|_{E_{uh}})<0.
\end{equation} 
$E_l$ is representative of a low cooperation regime in the entire network, $E_h$ is representative of a high cooperation regime in the entire network, and $E_m$ is representative of the regime where some lakes cooperate and high levels and some at low levels and are not qualitatively different from the initial condition of the network (i.e. $N-k$ lakes remain in a low cooperation state and $k$ lakes remain in a high cooperation state). We explore the loss of stability or disappearance of $E_m$. That is, when $E_m$ loses stability or vanishes the dynamics must either shift to $E_l$, or $E_h$ and the shift will be decided by the basins of attraction near bifurcation points. From graphical methods, we can show that when $E_m$ exists it is stable. However, $E_m$ can disappear through one of two saddle node bifurcations: (i) The equilibrium $E_{uh}$ collides with $E_m$ and
   (ii) The equilibrium $E_{ul}$ collides with $E_m$.
\subsubsection{Case (i): $E_{uh}$ collides with $E_m$. }\label{sec:casei}
In this case, we see that the lower two branches of the $F_l$ nullcline would be to the left of the rightmost branch of the $F_h$ nullcline as depicted in Figure~\ref{fig:ppEh}. The separatrix is approximately the middle branch of the $F_h$ nullcline and runs vertically through $E_{ul}$. With initial conditions such that $k$ lakes are in a high cooperation state and $N-k$ lakes in a low cooperation state the network dynamics will eventually tend to $E_h$ based on the attraction basins. The phase plane of this situation is shown in Figure~\ref{fig:ppEh}.
\begin{figure}
    \centering
    \includegraphics[width=0.5\paperwidth]{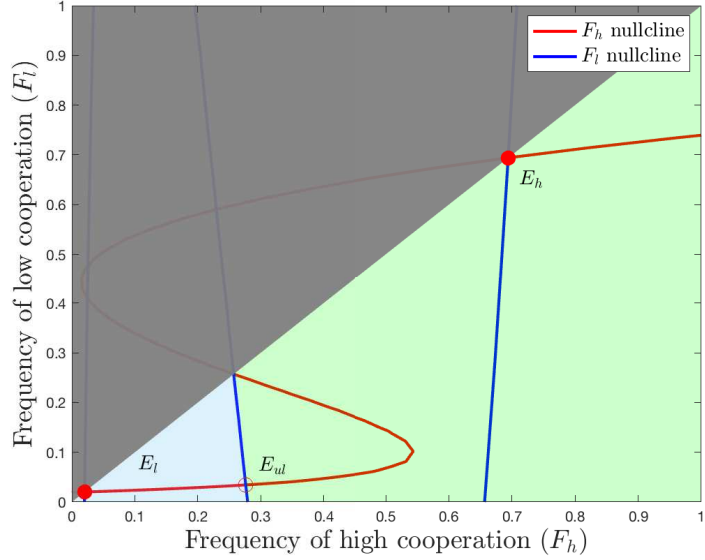}
    \caption{Phase plane after $E_m$ collides with $E_{uh}$. This phase plane corresponds with higher pressure to cooperate. The region above the line $F_l=F_h$ is excluded based on the condition $F_l<F_h$. The shaded regions represent the attraction basins of the stable equilibrium, the initial condition is located within the green region. In this case a regime shift to high cooperation occurs based on the prescribed initial condition.  }
    \label{fig:ppEh}
\end{figure}
\subsubsection{Case (ii): $E_{ul}$ collides with $E_m$.}\label{sec:caseii}

Here the two rightmost branches of the $F_l$ nullcline will be above the lowest branch of the $F_h$ nullcline. The separatrix is now the (roughly) middle branch of the $F_h$ nullcline and runs horizontally through $E_{uh}$. With initial conditions such that $k$ lakes are in a high cooperation state and $N-k$ lakes in a low cooperation state the network dynamics will eventually tend to $E_l$ based on the attraction basins. The phase plane of this situation is shown in Figure~\ref{fig:ppEl}

\begin{figure}
    \centering
    \includegraphics[width=0.5\paperwidth]{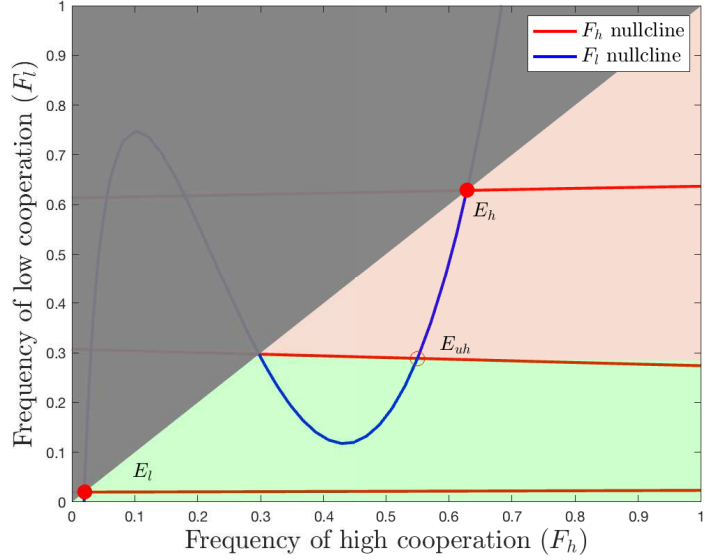}
    \caption{Phase plane after $E_m$ collides with $E_{ul}$. This phase plane corresponds with lower pressure to cooperate. The region above the line $F_l=F_h$ is excluded based on the condition $F_l<F_h$. The shaded regions represent the attraction basins of the stable equilibrium, the initial condition is located within the green region. In this case a regime shift to low cooperation occurs based on the prescribed initial condition. }
    \label{fig:ppEl}
\end{figure}

We have shown via phase plane analysis that three main dynamical outcomes occur. First, when additional network pressure is small then the system will stay in a state of mixed regime, i.e. $k$ lakes in high cooperation state and $N-k$ lakes in a low cooperation state as in Figure~\ref{fig:ppEm}. Second, when additional network pressure adds sufficiently more costs for the defectors a bifurcation occurs such that all lakes will tend to a high cooperation state as in Figure~\ref{fig:ppEh}. Lastly, when additional network pressures add sufficiently more costs to the cooperators a bifurcation occurs such that all lakes will tend a low cooperation state as in Figure~\ref{fig:ppEl}.  
\subsection{Bifurcation conditions}
We now discuss necessary conditions for the bifurcations to occur. These necessary conditions are based on model parameters and thus offer insight to the parameter values that result in certain regime shifts giving insight to effective mitigation for environmentally favourable outcomes. 
 
 Graphically we notice that the signs of $A|_{E_{m}}$ do not change as we vary parameter values. However, near either bifurcation point we require the signs $A|_{E_{ul}}$ or $A|_{E_{uh}}$ to change. This observation leads us to the following two theorems. 
 
 \begin{remark}\label{thm:networkbifurc}
  Necessary conditions for each bifurcation to occur are highlighted as follows: 
 \begin{itemize}
     \item[(i)] If system \eqref{eq:2Dnetwork} is close to the bifurcation point where $E_{ul}$ collides with $E_m$ then $\dfrac{df_1}{dF_h}|_{E_{ul}}<0$.
     \item[(ii)] If system \eqref{eq:2Dnetwork} is close to the bifurcation point where $E_{uh}$ collides with $E_m$ then $\dfrac{df_2}{dF_l}|_{E_{uh}}<0$. 
 \end{itemize}

 \end{remark}
 Remark \ref{thm:networkbifurc} is justified by graphical inspection. First, when sufficiently far away from the bifurcation point discussed in case (i) the Jacobian, $A|_{E_{ul}}=\begin{pmatrix}
 + & + \\
 + & - \end{pmatrix}$. A necessary condition at the bifurcation point is  $det(A|_{E_{ul}})=0$. Now, at the bifurcation point we require $A|_{E_{ul}}=A|_{E_{m}}$, and graphically we can see that near the bifurcation point only the first entry in $A|_{E_{ul}}$ will change sign. Thus, near the bifurcation point $\dfrac{df_1}{dF_h}|_{E_{ul}}$ must be negative. The argument for case (ii) follows identical logic to the above discussion. Thus, by Remark \ref{thm:networkbifurc} we have necessary conditions for the bifurcations to occur. These conditions involve model parameters, thus offering insight towards system characteristic that promote, or prevent, bifurcation from occurring.

\subsection{Bifurcation diagrams for the reduced network}

We now explore the possible regime shifts dependent on the socio-economic parameters $\delta_\coop,\delta_\defe$, $k$ and $N$ and focus on the bifurcations related to $E_m$. The model parameters considered offer insight towards implementing additional costs as to prevent non-environmentally favourable outcomes or to promote environmentally favourable regime shifts.

Figures~\ref{fig:2parbif},~\ref{fig:2parbifdcdd}, and~\ref{fig:2parbifD}, show two parameter bifurcation diagrams for the combination of parameters $\delta_\coop$, $\delta_\defe$ and $k/N$. We assume that the initial condition of~\eqref{eq:2Dnetwork} corresponds to populations in the bistable state of the single lake model i.e., $F_h(0)=F_h^*$ and $F_l(0)=F_l^*$. The region denoted with `high coop.' represents the region in parameter space where lake populations will tend to a high cooperation regime (Figure~\ref{fig:ppEh}). Conversely, the region denoted with `low coop.' corresponds to the network shifting to a low cooperation regime (Figure~\ref{fig:ppEl}). The region denoted `mixed coop.' represents when there is no regime shift (Figure~\ref{fig:ppEm}). There are two remaining regions in the parameter space where either $E_l$ or $E_h$ disappear. These regions correspond to the bifurcation that would occur when either i) $E_{ul}$ and $E_l$ collide, and ii) $E_{uh}$ and $E_h$ collide and are denoted with `high coop. only', and `low coop. only', respectively. In both cases the resultant regime will not change but the system becomes monostable instead of bistable offering the potential for global attractiveness of the regime. 
\begin{figure}
    \centering
    \includegraphics[width=\textwidth]{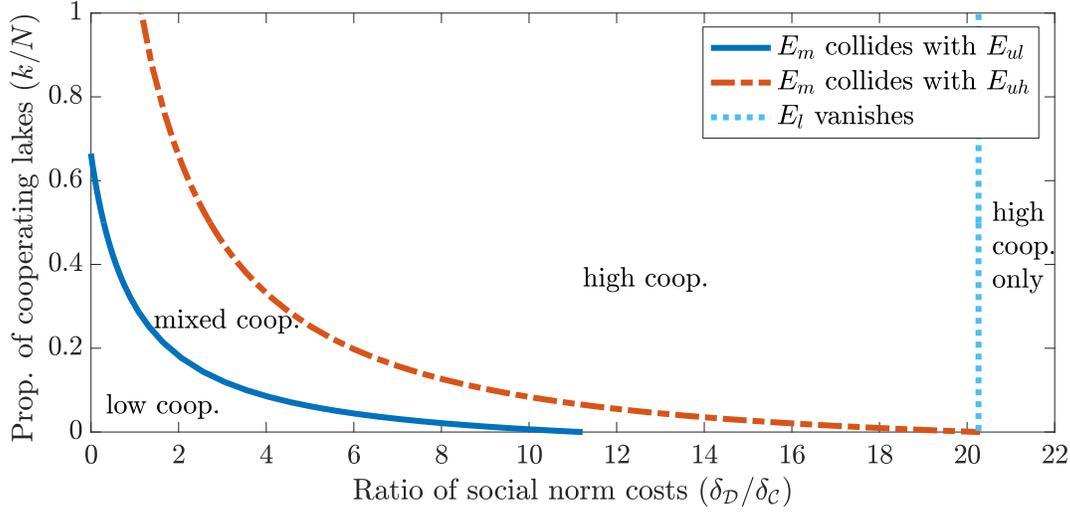}
    \caption{Shows the two parameter bifurcation for $\delta_\coop=0.6$. The solid line is the curve in parameter space where $E_m$ collides with $E_{ul}$. Below the solid is described in Case (ii). The dashed line is the curve in parameter space where $E_m$ collides with $E_{uh}$. Above the dashed line is described in Case (i). The region between the curves is where the equilibrium $E_m$ persists. To the right of the dotted line $E_l$ vanishes and the equilibrium $E_h$ is the only equilibrium.}
    \label{fig:2parbif}
\end{figure}

\begin{figure}
    \centering
    \includegraphics[width=\textwidth]{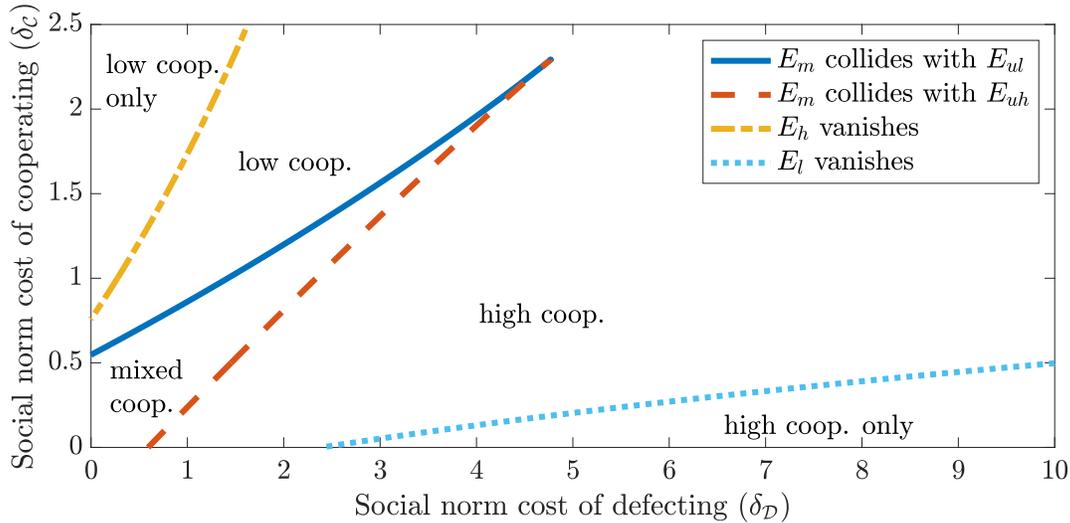}
    \caption{Shows the two parameter bifurcation for $k/N=0.5$. The orange dashed line is the curve in parameter space where $E_m$ collides with $E_{uh}$. The region below this line is described in Case (i). The solid blue line is the curve in parameter space where $E_m$ collides with $E_{ul}$. The region above this line is described in Case (ii). The solid and dashed lines meet at a cusp bifurcation (see Remark~~\ref{rem:cuspbifurc}). The other two lines show when $E_l$ or $E_h$ vanish. Crossing these lines transition from bistable state to a monostable state. }
    \label{fig:2parbifdcdd}
\end{figure}

\begin{remark}\label{rem:cuspbifurc}
There is a cusp bifurcation shown in Figures~\ref{fig:2parbifdcdd} and~\ref{fig:2parbifD}. The cusp bifurcation is the point in which $E_m$,$E_{ul}, $ and $E_{uh}$ all collide and one unstable equilibrium remains. Past the cusp bifurcation no bifurcation occurs, but the location of the unstable equilibrium varies and the size of the attracting basins changes accordingly. Thus, past the cusp bifurcation we label the region `high coop.' if remaining unstable equilibrium is along the bottom-most branch of the $F_l$ nullcline, and label the region `low coop.' if it is along the rightmost branch of the $F_h$ nullcline. Figure~\ref{fig:cuspbif} shows two phase portraits near the cusp bifurcation.
\end{remark}
 
\begin{figure}
    \centering
    \includegraphics[width=\textwidth]{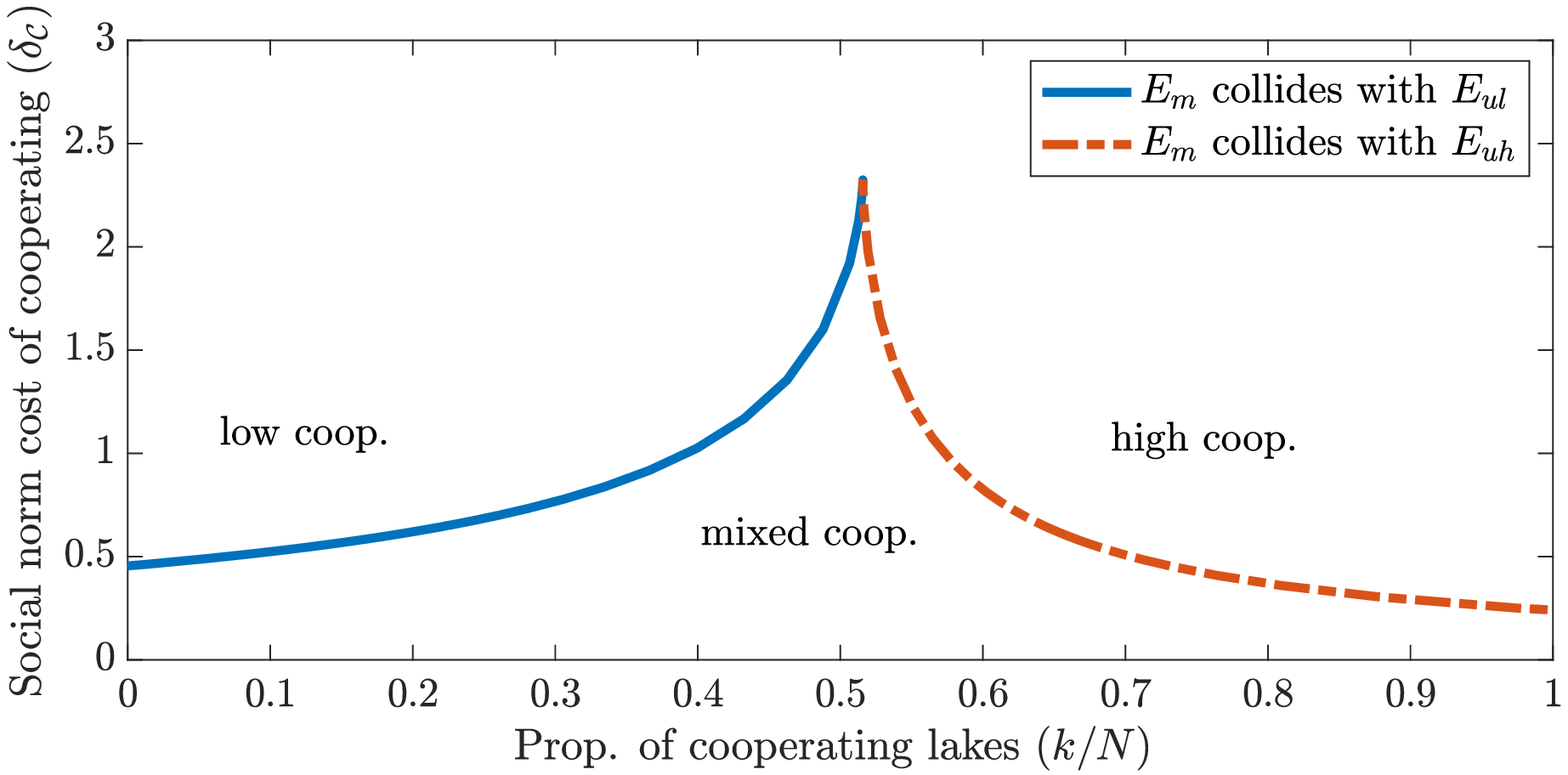}
    \caption{Shows the two parameter bifurcation for $\delta_\defe=2\delta_\coop$. The solid line is the curve in parameter space where $E_m$ collides with $E_{ul}$. Above the solid is described in Case (ii). The dashed line is the curve in parameter space where $E_m$ collides with $E_{uh}$. Above the dashed line is described in Case (i).The solid and dashed lines meet at a cusp bifurcation (see Remark~~\ref{rem:cuspbifurc}).}
    \label{fig:2parbifD}
\end{figure}

\begin{figure}
    \centering
    \includegraphics[width=0.49\textwidth]{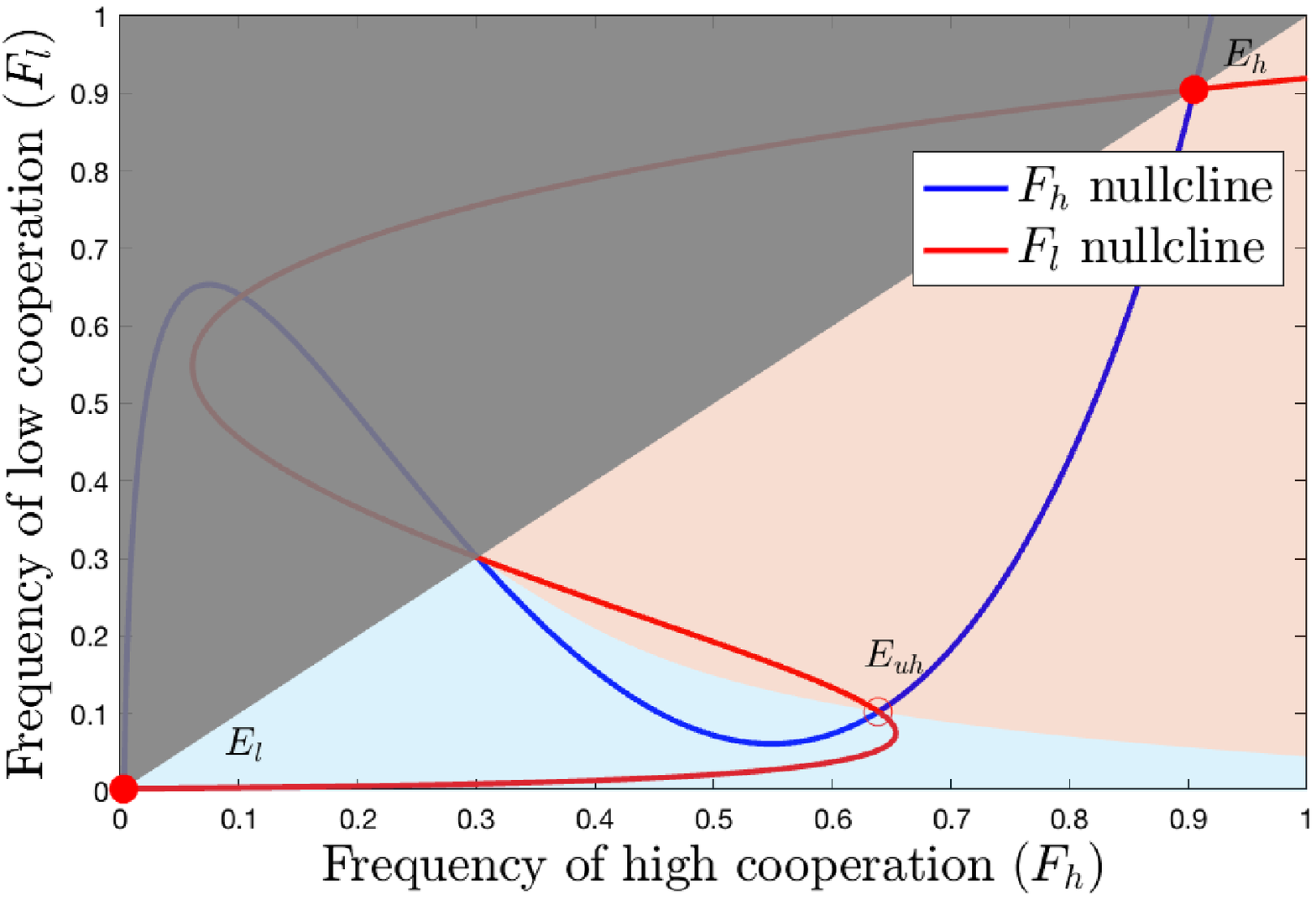}
    \includegraphics[width=0.49\textwidth]{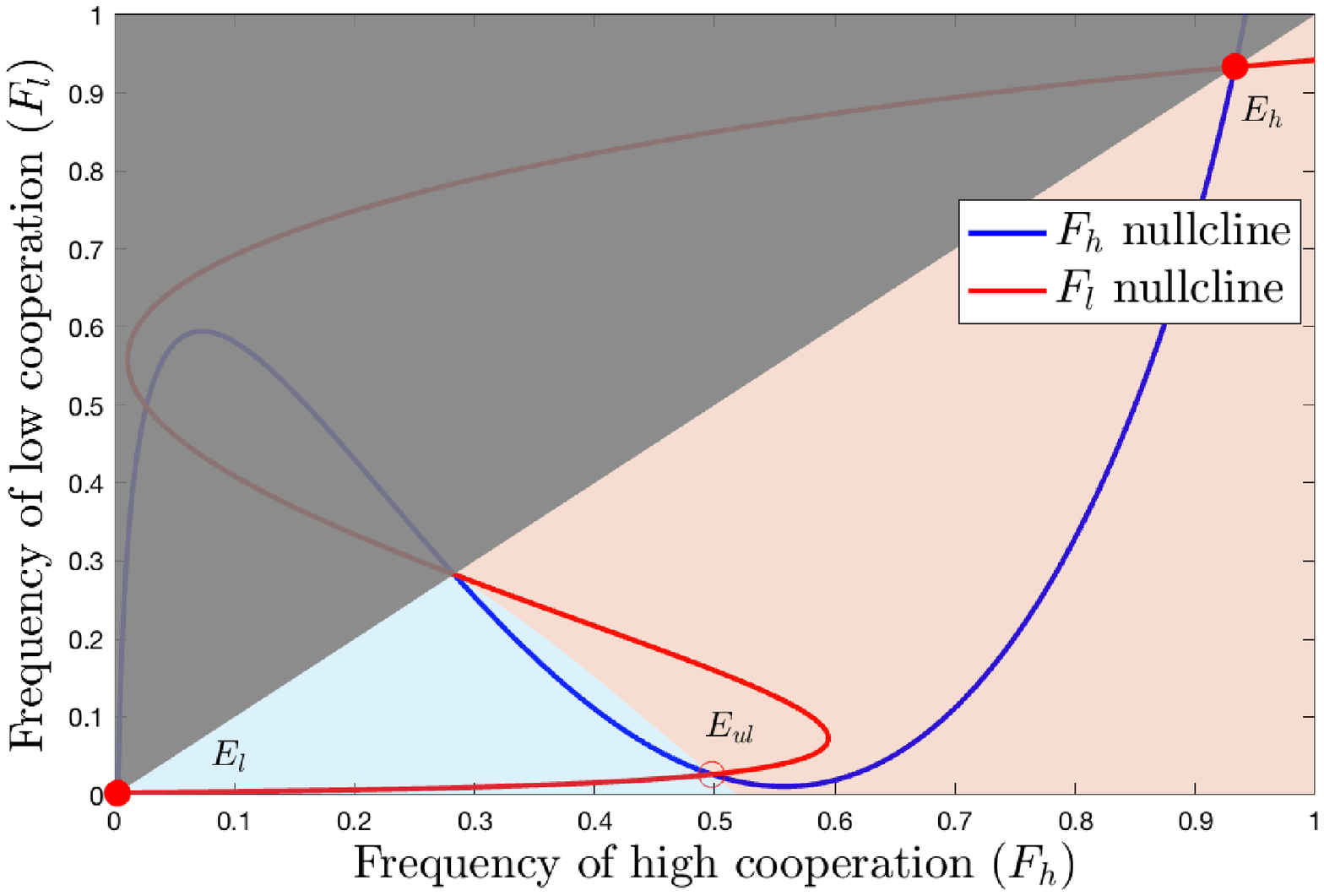}
    \caption{Shows the phase portrait just beyond the cusp bifurcation in Figures~\ref{fig:2parbifdcdd} and~\ref{fig:2parbifD}. The left phase portrait shows the unstable equilibrium on the rightmost branch of the $F_h$ nullcline, thus this parameter region is labelled as `low coop'. The right phase portrait shows the unstable equilibrium on the bottom-most branch of the $F_l$ nullcline, thus is labelled as the `high coop.' region. }
    \label{fig:cuspbif}
\end{figure}

The series of two parameter bifurcation plots presented in Figures~\ref{fig:2parbif},~\ref{fig:2parbifdcdd}, and~\ref{fig:2parbifD} give insight to parameter values that will yield a favourable regime shift. Interestingly, the parameters $\delta_\coop$, $\delta_\defe$, $k$ and $N$ can be altered strategically with three possible outcomes. First, for values that greatly increase the pressure to cooperate, such as large $k/N$, small $\delta_\coop$, or large $\delta_\defe$ the network will shift to a high cooperation regime. Second, for small  $k/N$, large $\delta_\coop$, or small $\delta_\defe$ ($\delta_\defe/\delta_\coop$ is small) the network will shift to a low cooperation regime. However, for the third case, there are intermediate values of each of the parameters such that no regime shift will occur.

\section{Discussion}

The study of CB dynamics is important to effectively mitigate potential risks associated with toxin production and ecosystem health as well as to reduce the associated agricultural, recreational, and water treatment costs. However, CB dynamics are intertwined with human dynamics through anthropogenic nutrient pollution~\cite{Paerl2014}. In order to make meaningful management strategies to reduce the effects of CHABs we must also consider the associated socio-economic dynamics.

In this paper we study the coupled socio-economic dynamics and abundance of CB in a single lake, and a network of lakes. The model presented is an extension of previously studied socio-economic and phytoplankton models. Similar systems have been studied but only consider the nutrient dynamics and neglect the deeper issue of CB abundance~\cite{Iwasa2007,Iwasa2010}. This distinction is important as many landowners are more concerned with the risks associated with CHABs as a result of eutrophication than the eutrophic conditions themselves. For this reason we explicitly consider the influence human dynamics have on CB abundance. This is done by extending the models established by~\textcite{Iwasa2007,Iwasa2010} to consider CB abundance as the main environmental concern. The CB model is derived from a series of stoichiometric models~\cite{Wang2007,Berger2006,Heggerud2020}. The coupling of the ecological and socio-economic models yields the existence of multiple stable states and hysteresis creating deeper implications for effective management of such systems. 

We consider the dynamics of our model for both phosphorus and iron. In the case of phosphorus the ecological dynamics occur on a faster timescale than the human dynamics. This observation allows us to apply the QSSA, simplifying the analysis. We further apply a series of approximations and arrive at a single tractable differential equation that describes the entire single lake system. Equilibria and their stability are studied through a bifurcation analysis with respect to the parameter $\eta$ in Theorem~\ref{thm:phase lineEQ}. The results are supported by graphical inspection, inspired by the analysis of the classical Spruce Budworm model~\cite{Ludwig1978}. For the iron case, the nutrient dynamics occur on the same timescale as the human dynamics with the CB and cell quota remaining on the faster timescale. We again apply the QSSA and graphically study the system in the phase plane. The graphical results show two saddle node bifurcation points with respect to the parameter $\eta$. Each bifurcation point results in a loss of bistability and either the high or low cooperation equilibrium become globally attracting. This type of dynamic is akin to the typical hysteresis phenomenon~\cite{Carpenter2005,Beisner2003}.

Lastly, we extend the phase line analysis to study the long term behaviour of a network of lakes. We assume that each lake is ecologically similar and that the only connections among lakes are through social interactions. The analysis is done in the phase plane where multiple stable states exist. Each state corresponds to a regime of high, low or mixed levels of cooperation throughout the network. Two main bifurcation branches are observed which correspond to the loss of the mixed cooperation equilibrium. Through a series of bifurcation diagrams we gain understanding as to what the long term regime outcome is based on parameter values pertaining to social pressures.

The results presented in this paper all have various implications towards policy and management strategies. In the phase line analysis of the phosphorus model our main result, Theorem~\ref{thm:phase lineEQ}, gives analytical conditions for bifurcation points. The bifurcation points correspond to the loss of bistability as the parameter $\eta$ is changed. Recall that $\eta$ represents the difference in baseline costs to cooperate and defect plus the difference in costs of external social pressures. That is, a large $\eta$ represents a larger cost to cooperate, whereas a small $\eta$ represents a larger cost to defect. Our results are not surprising in the sense that large $\eta$ leads to lower cooperation and vice versa, but what is noteworthy to managers is the presence of hysteresis. A lake could be in a high cooperation state and suddenly shift to a low cooperation state if $\eta$ exceeds its bifurcation point ($\eta_3$). However, attempts to lower $\eta$ back down to $\eta_3$ will not be sufficient in re-achieving the high cooperation state due to the presence of the hysteresis phenomenon. 

The results of the iron system reiterate the conclusions of the phosphorus system with respect to bistability, but explicitly show the high and low pollution states. Additionally, the dynamics of iron and phosphorus are assumed to act on different timescales with respect to CB and anthropogenic inputs~\cite{Whitton2012,Cunningham2017}. We assume that phosphorus dynamics occur on a similar scale to the ecological dynamics, whereas iron dynamics occur on the slower time scale similar to the human dynamics. This distinction is important when suggesting management strategies for a specific nutrient as the transient dynamics of the systems can differ significantly and moreover, the response of the human system may not yield a satisfactory response in the ecology for significantly longer periods of time~\cite{Hastings2016}. 

In general, adding costs that are associated with social norm pressures on the defectors can help sway the long term outcomes to be environmentally favourable. For instance, a large associated social norm pressure to cooperate, combined with a low associated pressure to defect will result in an overall environmentally favourable outcome. Also, when social network connections are added the initial state of the network can be a predictor of the regime outcome. We show that when a large majority of lakes start in a high cooperation state, it is unlikely that parameters can be changed enough so that the long term outcome will be a state of low cooperation, although mixed levels of cooperation throughout the network is possible.

We have shown that bistability is observed in a single lake system. These results reiterates what has been hypothesized to occur in many nutrient explicit lake systems~\cite{Carpenter2005,Iwasa2010} although in our case bistability does not occur in the CB model without the socio-economic coupling~\cite{Wang2007,Heggerud2020}. However, when a network system is considered our results show that tristability occurs as in Figure \ref{fig:ppEm}. To our knowledge no such ecological system has been shown to exhibit such behaviour. The implication of tristability is interesting in the sense that three regime outcomes are possible and that attempts to shift regimes may require significantly more effort. That is, if the system is in a low cooperation regime, it must first transition to a mixed regime state before achieving the environmentally favourable outcome. Furthermore, a system in the mixed regime can be perturbed in either direction to cause a regime shift, as opposed to the bistable case where perturbations can only shift the regime in one direction. In this sense tristable systems are more fragile to environmental fluctuations which can both be beneficial if transitioning from a less favourable regime, or detrimental if in a favourable state. Future work of this study should include deeper consideration of tristable systems and their implications towards management.

In much of our analysis we make the QSSA or similar simplification. Although these simplifications make the model tractable for analysis they do take away some key aspects of the dynamics, mainly the possibility for interesting transient dynamics. Our results pertain to only long-term dynamics which may be insufficient in the eyes of policy makers as the ecology can change drastically on a smaller timescale~\cite{Heggerud2020,Hastings2010,Hastings2018}. Furthermore, certain mechanisms are deemed negligible via the QSSA which, although reasonable, do take away from the dynamics of the full system.

The socio-economic component of our model uses the logit best-response dynamics to model human decision making and is extended from~\textcite{Iwasa2007,Iwasa2010}. We present justifications for using this form, although the replicator dynamics are, perhaps, more commonly used. Indeed the replicator dynamics are more mathematically friendly, but they are based on the assumption of completely rational individuals~\cite{Sun2021}, or individuals that base their decision solely on decreasing cost by assuming a strategy that was beneficial to another individual. The best-response dynamics assumes that an individual bases their decision partly on the current environmental state and the associated social norms~\cite{Farahbakhsh2021,Sun2021}. Furthermore, replication dynamics assumes that eventually the population will entirely assume a strategy whereas the best-response dynamics will have persistence of both strategies. For our system we deem the assumptions around the best-response dynamics more reasonable. However, for deeper mathematical understanding the simpler form of the replicator dynamics may prove useful. 

Our analysis of the network model is greatly simplified with the strict assumption that all lakes had identical ecological dynamics and socio-economic parameters. This assumption may be inaccurate in reality, but allows for a simplistic understanding of the potential regimes, and shifts among them. We assume the social norm pressures from the network are dependent on the proportion of cooperators at each lake, when in reality this pressure is also dependent on the population size at each lake. Additionally, explicit weighted network connections can result in coupling between pairs of lakes in which we expect to see scenarios where regime shifts can propagate through the network~\cite{Keitt2001}. In future extensions of our model these assumptions should be revisited and addition of weighted network connections that are non-uniform should be considered. By relaxing our assumptions on the network model we expect many new and exciting results pertaining to social dynamics and propagation of regime shifts. 

The current study shows the importance of the interconnection of ecological and socio-economic dynamics in aquatic systems by portraying the various dynamical outcomes that can occur in the coupled system. Social pressures and ostracism influence the role an individual assumes with respect to environmental issues by adding associated costs. Furthermore, social pressures can lead to favourable regime shifts within a network of lakes giving valuable insight to policies and mitigation strategies. This study builds a valuable framework for future studies of coupled CB and socio-economic systems.

\section*{Acknowledgements}
The authors would like to thank Professor Rolf Vinebrooke for helping motivate this paper and for many stimulating biological discussions. C.M.H. acknowledges financial support from the University of Alberta and the Alberta Graduate Excellence Scholarship (AGES), M.A.L. by a Canada Research Chair and both M.A.L. and H.W. by NSERC discovery grants.

\printbibliography
\appendix
\section{Appendix}\label{Appendix}
\begin{proof}[Proof of Theorem~\ref{thm:phase lineEQ}]
First, we note that $\hat\eta_1<\hat\eta_2<\hat\eta_3$ and $J(F)>0$ for all $F\in(0,1)$. Assume that $\hat\eta<\hat\eta_1<-1/2$ then $F+\hat\eta-1/2<0$ for all $F\in[0,1]$ and $J(F)>F+\hat\eta-1/2$ for all $F\in(0,1)$. Hence, $\dfrac{dF}{d\tilde\tau}>0$ for all $F\in(0,1)$ and $\dfrac{dF}{d\tilde\tau}=0$ for $F=1$, thus proving (i). 

 When $\hat\eta=\hat\eta_1$, $F^*_h=F^*_1=1$. Since $\frac{dF^*_h}{d\hat\eta}<0$, and necessary condition for $F^*_h\in[0,1]$ is $\hat\eta>\hat\eta_1$. Now, assume that $\hat\eta_1<\hat\eta<\hat\eta_2$, then $F^*_u<0$ because $F^*_u=F^*_l=0$ when $\hat\eta=\hat\eta_2$ and $\frac{dF^*_h}{d\hat\eta}>0$. Moreover, $J(F)=F+\hat\eta-1/2$ has only one solution for $F\in[0,1]$ then by the concavity of $J(F)$:
\begin{equation}
    J''(F)=-\frac{2\,a_{1}\,\sigma \,\left(a_{2}+\xi +a_{2}\,\xi \right)}{{\left(a_{2}(1-F)+1\right)}^3}<0,
\end{equation} for all $F$,
$J(F)>F+\hat\eta-1/2$ for all $F\in[0,F^*_h)$ and $J(F)<F+\hat\eta-1/2$ for all $F\in(F^*_h,1]$. Hence, $J(F)+1/2-\hat\eta>0$ for all $F\in[0,F^*_u)$ and $J(F)+1/2-\hat\eta<1$ for all $F\in(F^*_u,1]$. This implies that $\frac{dF^*_h}{d\hat\eta}>0$ for all $F\in[0,F^*_h)$ and $\frac{dF^*_h}{d\hat\eta}<0$ for all $F\in(F^*_h,1]$ thus proving (ii).

When $\hat\eta=\hat\eta_3$, $F^*_h=F^*_u$ as seen in~\eqref{eq:Fu} and~\eqref{eq:Fl}. Now assume that $\hat\eta_2<\hat\eta<\hat\eta_3$, then $1>F^*_h>F^*_u>0$ by their respective monotonicity. Since $J(F)=F+\hat\eta-1/2$ for $F=F^*_h$ and $F=F^*_u$, we deduce that $0<J(F)+1/2-\hat\eta<1$ near $F^*_h$ and $F^*_u$. Furthermore, by the concavity of $J(F)$,  $J(F)>F+\hat\eta-1/2$ for all $F\in(F^*_u,F^*_h)$ and $J(F)<F+\hat\eta-1/2$ for all $F\notin[F^*_u,F^*_h]$ thus proving that $F^*_h$ is locally stable, and $F^*_u$ is unstable. Lastly, if $F<F^*_u$ then $J(F)<F+\hat\eta-1/2$. Which furthermore implies that $J(F)-\hat\eta+1/2 \leq 0$ for $F$ near $F=0$, thus proving that $F=0$ locally stable and concluding (iii)

Finally, assume $\hat\eta>\hat\eta_3$, then both $F^*_h$ and $F^*_u$ are imaginary roots and not considered to be equilibrium. Furthermore, $J(F)<F+\hat\eta-1/2$ for all $F\in[0,1]$ and hence, $J(F)-\hat\eta+1/2<1 $ for all $F\in[0,1]$ implying that $\frac{dF}{d\tilde\tau}<0$ for all for all $F\in(0,1]$ and $\frac{dF}{d\tilde\tau}=0$ for $F=0$ proving that $F=F^*_l$ is globally stable concluding (iv). 
\end{proof}

\begin{proof}[Proof of Corollary~\ref{cor:saddlenode}]
The two equilibria $F^*_h$ and $F^*_u$, which are stable and unstable respectively collide, are equivalent at $\hat{\eta}=\hat{\eta}_3$ and do not exist for $\hat{\eta}>\hat{\eta}_3$. Furthermore, when $\hat\eta=\hat\eta_3$, $0<J(F)-\hat\eta+1/2<1$ near $F^*_u=F^*_h$, since  $J(F)-\hat\eta+1/2=F$ and $\frac{dF}{d\tilde\tau}=0$ at $F^*_u=F^*_h$. Lastly, we check that $\frac{d(J(F)+1/2-\hat\eta-F)}{dF}=0$ which is equivalent to $J'(F)=1$ at $F_h^*$ and $\eta=\eta_3$.

Recall $F_h^*$ as given in~\eqref{eq:Fl} and $\hat{\eta}_3$ as in~\eqref{eq:etaroots}. Then at $\hat{\eta}=\hat{\eta}_3$:
\begin{equation}
    F_h^*=F_l^*=    -\frac{\splitdfrac{\sqrt{-a_{1}\,\sigma \,\left(a_{2}-a_{1}\,\sigma \,\xi \right)\,\left(a_{2}+\xi +a_{2}\,\xi \right)}}{-a_{2}-{a_{2}}^2+a_{1}\,\sigma \,\xi +a_{1}\,a_{2}\,\sigma \,\xi} }{{a_{2}}^2-a_{1}\,a_{2}\,\sigma \,\xi }.
\end{equation}
Thus, via substitution and tedious computations that are verified using MATLAB's symbolic software

\begin{align}
   J'(F_h^*)&=    \frac{a_{1}\,\sigma \,\left(\xi -2\,F_h^*\,\xi +a_{2}\,\xi -2\,F_h^*\,a_{2}\,\xi +F_h^{*^2}\,a_{2}\,\xi -1\right)}{{\left(a_{2}-F_h^*\,a_{2}+1\right)}^2},
   \\[1em]
   &=\dfrac{a_{1}\,\sigma \,\left(  
   \splitdfrac{\splitdfrac{\splitdfrac{{}\xi +a_{2}\,\xi-1}{{} +\frac{2\,\xi \,\left(\sqrt{-a_{1}\,\sigma \,\left(a_{2}-a_{1}\,\sigma \,\xi \right)\,\left(a_{2}+\xi +a_{2}\,\xi \right)}-a_{2}-{a_{2}}^2+a_{1}\,\sigma \,\xi +a_{1}\,a_{2}\,\sigma \,\xi \right)}{{a_{2}}^2-a_{1}\,a_{2}\,\sigma \,\xi }}
   }
   {
    +\frac{2\,a_{2}\,\xi \,\left(\sqrt{-a_{1}\,\sigma \,\left(a_{2}-a_{1}\,\sigma \,\xi \right)\,\left(a_{2}+\xi +a_{2}\,\xi \right)}-a_{2}-{a_{2}}^2+a_{1}\,\sigma \,\xi +a_{1}\,a_{2}\,\sigma \,\xi \right)}{{a_{2}}^2-a_{1}\,a_{2}\,\sigma \,\xi }}
    }
    {
    +\frac{a_{2}\,\xi \,{\left(\sqrt{-a_{1}\,\sigma \,\left(a_{2}-a_{1}\,\sigma \,\xi \right)\,\left(a_{2}+\xi +a_{2}\,\xi \right)} -a_{2}-{a_{2}}^2+a_{1}\,\sigma \,\xi +a_{1}\,a_{2}\,\sigma \,\xi \right)}^2}{{\left({a_{2}}^2-a_{1}\,a_{2}\,\sigma \,\xi \right)}^2}
    }\right)}
    {{\left(a_{2}+\frac{a_{2}\,\left(\sqrt{-a_{1}\,\sigma \,\left(a_{2}-a_{1}\,\sigma \,\xi \right)\,\left(a_{2}+\xi +a_{2}\,\xi \right)}-a_{2}-{a_{2}}^2+a_{1}\,\sigma \,\xi +a_{1}\,a_{2}\,\sigma \,\xi \right)}{{a_{2}}^2-a_{1}\,a_{2}\,\sigma \,\xi }+1\right)}^2}
    \\
    &=1.
\end{align}

\end{proof}

\begin{proof}[Proof of Corollary~\ref{cor:saddlenodezero}]
  In~\eqref{eq:J(F)equalseta} we see that $F^*_u$ and $F^*_l$ are equivalent (collide) at $\eta=\eta_2$ and are equal to zero. Furthermore, for $\eta<\eta_2$ neither steady state exists as $F^*_u<0$ as $F^*_u$ is a increasing function of $\eta$ and $0=J(F_l^*)+1/2-\eta_2$ implies that $J(F_l^*)+1/2-\eta>0$ for $\eta<\eta_2$. Thus the point $(F,\hat\eta)=(F^*_l,\hat\eta_2)$ is a saddle node bifurcation.  
\end{proof}

\end{document}